\documentclass[10pt]{amsart}
\usepackage{amsmath}
\usepackage{amssymb}
\usepackage{amsthm}
\usepackage{amsfonts}
\usepackage{amstext}
\usepackage{amsopn}
\usepackage{amsxtra}
\usepackage[latin1]{inputenc}
\usepackage{dsfont}
\usepackage[active]{srcltx}
\usepackage{hyperref}
\newcommand\R{{\ensuremath {\mathbb R} }}
\newcommand\C{{\ensuremath {\mathbb C} }}
\newcommand\N{{\ensuremath {\mathbb N} }}

\newcommand\1{{\ensuremath {\mathds 1} }}
\renewcommand\phi{\varphi}
\newcommand{\alp}{\boldsymbol{\alpha}}
\newcommand{\alphaph}{\alpha_{\rm ph}}
\newcommand{\rhoph}{\rho_{\rm ph}}
\newcommand{\gH}{\mathfrak{H}}

\newcommand{\gS}{\mathfrak{S}}

\renewcommand{\to}{\rightarrow}

\newcommand{\cN}{\mathcal{N}}

\newcommand{\cU}{\mathcal{U}}
\newcommand{\cB}{\mathcal{B}}

\newcommand{\cF}{\mathcal F}

\newcommand{\tr}{{\rm tr}\,}

\newcommand{\cC}{\mathcal{C}}
\newcommand{\cQ}{\mathcal{Q}}

\newcommand\ii{{\ensuremath {\infty}}}

\newcommand{\norm}[1]{ \left| \! \left| #1 \right| \! \right| }

\newtheorem{thm}{Theorem}
\newtheorem{lemma}{Lemma}
\newtheorem{corollary}{Corollary}
\newtheorem{prop}{Proposition}

\newtheorem{remark}{Remark}

\numberwithin{equation}{section}
 \numberwithin{lemma}{section}
\numberwithin{prop}{section}
\numberwithin{corollary}{section}
\renewcommand{\tr}{{\rm Tr} }

\long\def\symbolfootnote[#1]#2{\begingroup%
\def\thefootnote{\fnsymbol{footnote}}\footnote[#1]{#2}\endgroup}

\renewcommand{\leq}{\leqslant}
\renewcommand{\geq}{\geqslant}

\newtheorem{step}{Step}

\begin{document}

\title[Charge Renormalization]{Renormalization and asymptotic expansion of Dirac's polarized vacuum}

\author[P. Gravejat]{Philippe GRAVEJAT}
\address{Ceremade (UMR 7534), Universit{\'e} Paris-Dauphine, Place du Mar{\'e}chal de Lattre de Tassigny, F-75775 Paris Cedex 16, France.}
\email{gravejat@ceremade.dauphine.fr}

\author[M. Lewin]{Mathieu LEWIN}
\address{CNRS \& Laboratoire de Mathématiques (UMR 8088), Université de Cergy-Pontoise, F-95000 Cergy-Pontoise, France.}
\email{Mathieu.Lewin@math.cnrs.fr}

\author[\'E. S\'er\'e]{\'Eric S\'ER\'E}
\address{Ceremade (UMR 7534), Universit{\'e} Paris-Dauphine, Place du Mar{\'e}chal de Lattre de Tassigny, F-75775 Paris Cedex 16, France.}
\email{sere@ceremade.dauphine.fr}

\date{April 9, 2010}

\begin{abstract}
We perform rigorously the charge renormalization of the so-called reduced Bogoliubov-Dirac-Fock (rBDF) model. This nonlinear theory, based on the Dirac operator, describes atoms and molecules while taking into account vacuum polarization effects. 
We consider the total physical density $\rhoph$ including both the external density of a nucleus and the self-consistent polarization of the Dirac sea, but no `real' electron.  We show that $\rhoph$ admits an asymptotic expansion to any order in powers of the physical coupling constant $\alphaph$, provided that the ultraviolet cut-off behaves as $\Lambda\sim e^{3\pi(1-Z_3)/2\alphaph}\gg1$. The renormalization parameter $0<Z_3<1$ is defined by $Z_3=\alphaph/\alpha$ where $\alpha$ is the bare coupling constant. The coefficients of the expansion of $\rhoph$ are independent of $Z_3$, as expected. The first order term gives rise to the well-known Uehling potential, whereas the higher order terms satisfy an explicit recursion relation.

\bigskip

\noindent\scriptsize\copyright~2010 by the authors. This paper may be reproduced, in its entirety, for non-commercial~purposes.
\end{abstract}


\maketitle

\setcounter{tocdepth}{1}
\tableofcontents

\section{Introduction and main result}

Renormalization is an essential tool in Quantum Electrodynamics (QED) \cite{Dyson-49b,BjorkenDrell-65,ItzyksonZuber}. The purpose of this paper is to perform rigorously the charge renormalization of a nonlinear approximation of QED, the reduced Bogoliubov-Dirac-Fock (rBDF) theory that was studied before in \cite{HaiLewSer-05a,HaiLewSer-05b,HaiLewSol-07,HaiLewSerSol-07,HaiLewSer-08,GraLewSer-09}. This model, based on the Dirac operator, describes atoms and molecules while taking into account vacuum polarization effects. It does not need any mass renormalization, hence it is a theory simple enough for an investigation of charge renormalization in full detail.

Before turning to our specific Dirac model, let us quickly recall the spirit of renormalization. A physical theory usually aims at predicting physical observables in terms of the parameters in the model. Sometimes, interesting quantities are divergent and it is necessary to introduce cut-offs. For electrons the parameters are their mass $m$ and their charge $e$ (or rather the coupling constant $\alpha=e^2$). Predicted physical quantities are then functions $F(m,\alpha,\Lambda)$ where $\Lambda$ is the regularization parameter. Mass and charge are also physical observables and renormalization occurs when their values predicted by the theory are different from their `bare' values:
\begin{equation}
m_{\rm ph}=m_{\rm ph}(m,\alpha,\Lambda)\neq m\quad\text{and/or}\quad \alphaph=\alphaph(m,\alpha,\Lambda)\neq \alpha.
\label{eq:ph_mass_and_charge} 
\end{equation}
In this case the parameters $m$ and $\alpha$ are not observable in contrast with $m_{\rm ph}=m_{\rm ph}(m,\alpha,\Lambda)$ and $\alphaph=\alphaph(m,\alpha,\Lambda)$ which have to be set equal to their experimental values. The relation \eqref{eq:ph_mass_and_charge} has to be inverted, in order to express the bare parameters in terms of the physical ones:
\begin{equation}
m=m(m_{\rm ph},\alphaph,\Lambda)\qquad \alpha=\alpha(m_{\rm ph},\alphaph,\Lambda).
\label{eq:bare_mass_and_charge}
\end{equation}
This allows to express any observable quantity $F$ as a function $\tilde{F}$ of the physical parameters and the cut-off $\Lambda$:
\begin{equation}
\tilde F(m_{\rm ph},\alphaph,\Lambda)=F\big(m(m_{\rm ph},\alphaph,\Lambda)\,,\,\alpha(m_{\rm ph},\alphaph,\Lambda)\,,\,\Lambda\big)
\label{eq:ph_quantities}
\end{equation}
A possible definition of renormalizability is that \emph{all such observable quantities have a limit when $\Lambda\to\ii$, for fixed $m_{\rm ph}$ and $\alphaph$}.

Important difficulties can be encountered when trying to complete this program:
\begin{itemize}
 \item The physical quantities $m_{\rm ph}$ and $\alphaph$ might be \emph{nonexplicit functions of $\alpha$ and $m$}. The corresponding formulas can then only be inverted perturbatively to any order (usually in $\alpha$). This is the case in QED \cite{Dyson-49b,BjorkenDrell-65,ItzyksonZuber}. In the model studied in this paper we have $m_{\rm ph}=m$ and $\alphaph\neq \alpha$, hence only the charge has to be renormalized. Furthermore $\alphaph$ is an \emph{explicit function} of $m$, $\alpha$ and $\Lambda$ (see \eqref{charge_renormalization} later). Renormalizing our model is therefore a much easier task than in full QED.
\item Even when the bare parameters are explicit functions of the physical ones, these relations can make it \emph{impossible to take the limit $\Lambda\to\ii$ while keeping $m_{\rm ph}$ and $\alphaph$ fixed}. As we will explain, in our model $(2/3\pi)\alphaph\log\Lambda\leq1$. To deal with this problem, we let $\Lambda$ depend on $\alphaph$ and we investigate the asymptotics in the limit $\alphaph\to0$. 
\end{itemize}

We now turn to the description of our model. The Bogoliubov-Dirac-Fock theory is the Hartree-Fock approximation of QED when photons are neglected \cite{HaiLewSol-07,HaiLewSerSol-07}. The associated \emph{reduced} theory is obtained by further neglecting the so-called exchange term. In both models, the system is described by a Hartree-Fock (quasi-free) state in Fock space, which is completely characterized by its one-body density matrix $P$ (an orthogonal projector for pure states), acting on the one-body space. The state $P$ contains both the `real' electrons of the system (that of an atom for instance) and the `virtual' electrons of the Dirac sea, which all interact with each other self-consistently. Therefore, there are always infinitely many particles and $P$ is infinite-rank.

When the exchange term is neglected, a ground state at zero temperature is (formally) a solution of the following self-consistent equation:
\begin{equation}
\left\{\begin{array}{l}
P=\chi_{(-\infty,\mu)}\left(D\right)+\delta\\
D=D^0+\alpha(\rho_{P-1/2}-\nu)\ast|x|^{-1}.
\end{array}\right.
\label{eq:SCF_no_cut_off} 
\end{equation}
Here $D^0=\boldsymbol{\alpha}\cdot(-i\nabla)+\beta$ is the free Dirac operator \cite{Thaller} acting on the Hilbert space $\gH:=L^2(\mathbb{R}^2,\mathbb{C}^4)$. For the sake of simplicity we have chosen units in which the speed of light is $c=1$ and, as the model does not need any mass renormalization, we have taken $m=1$ for the mass of the electrons.
The second term in the formula of $D$ is the Coulomb potential induced by both a fixed external density of charge $\nu$ (modelling for instance a smeared nucleus) and the self-consistent density $\rho_{P-1/2}$ of the system (see below). In \eqref{eq:SCF_no_cut_off}, $\alpha$ is the \emph{bare} coupling constant that will be renormalized later and $\mu\in(-1,1)$ is a chemical potential which is chosen to fix the desired total charge of the system. We have added in \eqref{eq:SCF_no_cut_off} the possibility of having a density matrix $0\leq\delta\leq\chi_{\{\mu\}}(D)$ at the Fermi level, as is usually done in reduced Hartree-Fock theory \cite{Solovej-91}. So the operator $P$ is not necessarily a projector but we still use the letter $P$ for convenience. Later we will restrict ourselves to the case of $P$ being an orthogonal projector.

Equation \eqref{eq:SCF_no_cut_off} is well-known in the physical literature. A model of the same form (including an exchange term) was proposed by Chaix and Iracane in \cite{ChaIra-89}. Also, similar equations are found in relativistic Density Functional Theory, usually with additional empirical exchange-correlation terms and classical terms accounting for the interactions with photons, see, e.g., \cite[Eq. (6.2)]{EngDre-87} and \cite[Eq. (62)]{Engel-02}. Dirac already considered in \cite{Dirac-34b} the first order term obtained from \eqref{eq:SCF_no_cut_off} in an expansion in powers of $\alpha$.

Let us now explain the exact meaning of $\rho_{P-1/2}$. The charge density of an operator $A:\gH\to\gH$ with integral kernel $A(x,y)_{\sigma,\sigma'}$ is formally defined as $\rho_A(x)=\sum_{\sigma=1}^4A(x,x)_{\sigma,\sigma}=\tr_{\C^4}(A(x,x))$. In usual Hartree-Fock theory, the charge density is $\rho_P(x)$. However, as there are infinitely many particles, this does not make sense here. In \eqref{eq:SCF_no_cut_off}, the subtraction of half the identity is a convenient way to give a meaning to the density, independently of any reference. One has formally
$$\rho_{P-1/2}(x)=\rho_{\frac{P-P^\perp}2}(x)=\frac12\sum_{i\geq1}|\varphi_i^-(x)|^2-|\varphi_i^+(x)|^2$$
where $\{\varphi_i^-\}_{i\geq1}$ is an orthonormal basis of $P\gH$ and $\{\varphi_i^+\}_{i\geq1}$ is an orthonormal basis of $(1-P)\gH$. As was explained in \cite{HaiLewSol-07}, subtracting $1/2$ to the density matrix $P$ of the Hartree-Fock state makes the model invariant under charge conjugation.

When there is no external field, $\nu\equiv0$, Equation \eqref{eq:SCF_no_cut_off} has an obvious solution for any $\mu\in(-1,1)$, the Hartree-Fock state made of all electrons with negative energy:
$$P=P^0_-:=\chi_{(-\ii,0)}(D^0),$$
in accordance with Dirac's ideas \cite{Dirac-28,Dirac-30,Dirac-33}. Indeed $\rho_{P^0_--1/2}\equiv0$, as is seen by writing in the Fourier representation
$$(P^0_--1/2)(p)=-\frac{\boldsymbol{\alpha}\cdot p+\beta}{2\sqrt{1+|p|^2}}$$
and since the Dirac matrices are trace-less. This shows the usefulness of the subtraction of half the identity to $P$, since the free vacuum $P^0_-$ now has a vanishing density.  For a general state $P$, we can use this to write (formally):
\begin{equation}
\rho_{P-1/2}=\rho_{P-1/2}-\rho_{P^0_--1/2}=\rho_{P-P^0_-}. 
\end{equation}
When $P$ belongs to a suitable class of perturbations of $P^0_-$ (for instance when $P-P^0_-$ is locally trace-class), the density $\rho_{P-P^0_-}$ is a well-defined mathematical object. We will give below natural conditions which garantee that $P-P^0_-$ has a well-defined density in our context.

In the presence of an external field, $\nu\neq0$, Equation \eqref{eq:SCF_no_cut_off} has \emph{no solution} in any `reasonable' Banach space \cite{HaiLewSer-05b} and it is necessary to introduce an ultraviolet regularization parameter $\Lambda$. The simplest method (although probably not optimal regarding regularity issues \cite{GraLewSer-09}) is to impose a cut-off at the level of the Hilbert space, that is to replace $\gH$ by
$$\gH_\Lambda:=\{f\in L^2(\R^3;\C^4),\ {\rm supp}(\widehat{f})\subset B(0;\Lambda)\}$$
and to solve, instead of \eqref{eq:SCF_no_cut_off}, the regularized equation in $\gH_\Lambda$:
\begin{equation}
\left\{\begin{array}{l}
P=\chi_{(-\infty,\mu)}\left(D\right)+\delta\\
D=\Pi_\Lambda\left(D^0+\alpha(\rho_{P-P^0_-}-\nu)\ast|x|^{-1}\right)\Pi_\Lambda
\end{array}\right.
\label{eq:SCF_cut_off} 
\end{equation}
where $\Pi_\Lambda$ is the orthogonal projector onto $\mathfrak{H}_\Lambda$ in $\gH=L^2(\R^3;\C^4)$.

Existence of solutions to \eqref{eq:SCF_cut_off} was proved in \cite{HaiLewSer-05b} for $\mu=0$ and in \cite{GraLewSer-09} for $\mu\neq0$. The precise statement is the following\footnote{To be more precise, in \cite[Theorem 1]{GraLewSer-09}, only the existence and uniqueness of minimizers of the reduced BDF functional are stated. Elementary arguments based on convexity allow to deduce Theorem \ref{Thm:existence} from the results of \cite{GraLewSer-09}.}:
\begin{thm}[Existence of self-consistent solutions to \eqref{eq:SCF_cut_off}, \cite{HaiLewSer-05b,GraLewSer-09}]\label{Thm:existence}
Assume that $\alpha\geq0$, $\Lambda>0$ and $\mu\in[-1,1]$ are given. Let $\nu$ in the so-called \emph{Coulomb space}:
$$\mathcal{C}:=\left\{f\ :\ \int_{\R^3}|k|^{-2}|\widehat{f}(k)|^2dk<\infty\right\}.$$
Then, Equation \eqref{eq:SCF_cut_off} has at least one {solution} $P$ such that 
\begin{equation}
P-P^0_-\in\mathfrak{S}_2(\mathfrak{H}_\Lambda),\quad P^0_\pm(P-P^0_-)P^0_\pm\in\mathfrak{S}_1(\mathfrak{H}_\Lambda),\quad \rho_{P-P^0_-}\in \mathcal{C}\cap L^2(\R^3).
\label{properties}
\end{equation}
All such solutions share the same density $\rho_{P-P^0_-}$.
\end{thm}
In \eqref{properties}, $\mathfrak{S}_1(\mathfrak{H}_\Lambda)$ and $\gS_2(\gH_\Lambda)$ are respectively the spaces of trace-class and Hilbert-Schmidt operators \cite{Simon-79} on $\mathfrak{H}_\Lambda$, and $P^0_+=1-P^0_-$.
The method used in \cite{HaiLewSer-05b,GraLewSer-09} was to identify solutions of \eqref{eq:SCF_cut_off} with minimizers of the so-called \emph{reduced Bogoliubov-Dirac-Fock energy} which is nothing but the formal difference between the reduced Hartree-Fock energy of $P$ and that of the reference state $P^0_-$. Note that due to the uniqueness of $\rho_{P-P^0_-}$ the mean-field operator $D$ is also unique and only $\delta$ can differ between two solutions of \eqref{eq:SCF_cut_off}. 

Let us mention that it is natural to look for a solution of \eqref{eq:SCF_cut_off} such that $P-P^0_-$ is a Hilbert-Schmidt operator on $\mathfrak{H}_\Lambda$. If $P$ is a projector, the Shale-Stinespring theorem \cite{ShaSti-65} then tells us that $P$ yields a Fock representation equivalent to that of $P^0_-$. Even when $P$ is not a projector, it will be associated with a unique Bogoliubov mixed state in the Fock space representation of $P^0_-$. This is a mathematical formulation of the statement that $P$ should not be too far from $P^0_-$. Indeed, if $P$ is an orthogonal  projector, one has (see \cite[Lemma 2]{HaiLewSer-05a} and \cite[Lemma 1]{HaiLewSer-08})
$$\left.\begin{array}{r}
P-P^0_-\in\mathfrak{S}_2(\mathfrak{H}_\Lambda)\\
P^2=P\\
\end{array}\right\}
\Longrightarrow P^0_\pm(P-P^0_-)P^0_\pm\in\mathfrak{S}_1(\mathfrak{H}_\Lambda)\ \text{and}\  \rho_{P-P^0_-}\in \mathcal{C}\cap L^2(\R^3),$$
therefore, in this case, \eqref{properties} is just equivalent to the Shale-Stinespring condition $P-P^0_-\in\gS_2(\gH_\Lambda)$.

The property $P^0_\pm(P-P^0_-)P^0_\pm\in\mathfrak{S}_1(\mathfrak{H}_\Lambda)$ allows us to define the total `charge' of the system by (see \cite{HaiLewSer-05a})
$$\tr_{P^0_-}(P-P^0_-):=\tr\; P^0_-(P-P^0_-)P^0_-\;+\;\tr\; P^0_+(P-P^0_-)P^0_+.$$
When $P$ is a projector, the above quantity is always an integer which is indeed nothing but the relative index of the pair $(P,P^0_-)$, see \cite{HaiLewSer-05a,AvrSeiSim-94}. Varying $\mu$ allows to pick the desired total charge. Indeed, if $\nu$ is small enough and $\mu=0$, then one has $\norm{P-P^0_-}<1$ and the relative index vanishes: $\tr_{P^0_-}(P-P^0_-)=0$.

\medskip

It is very important to realize that solutions of \eqref{eq:SCF_cut_off} are singular mathematical objects. This fact is precisely at the origin of charge renormalization. In \cite[Theorem 1]{GraLewSer-09}, the following was proved:

\begin{thm}[Nonperturbative charge renormalization formula \cite{GraLewSer-09}]\label{Thm:renorm}
Assume that $\alpha\geq0$, $\Lambda>0$ and $\mu\in(-1,1)$ are given. If $\nu\in\cC\cap L^1(\R^3)$, then $\rho_{P-P^0_-}\in L^1(\R^3)$ and it holds 
\begin{equation}
\int_{\R^3}\nu -\int_{\R^3}\rho_{P-P^0_-}=\frac{\displaystyle\int_{\R^3}\nu-\tr_{P^0_-}(P-P^0_-)}{1+\alpha B_\Lambda}.
\label{relation_charge}
\end{equation}
\end{thm}

In this formula, $B_\Lambda$ is an explicit function of the ultraviolet cut-off $\Lambda$ (see the comments after \eqref{expression_B} and \eqref{def:BL}), which behaves like
$$B_\Lambda=\frac{2}{3\pi}\log\Lambda-\frac{5}{9\pi}+\frac{2\log2}{3\pi}+O(1/\Lambda^2).$$
Let us emphasize that \eqref{relation_charge} is non perturbative and holds for all $\alpha\geq0$ and all $\mu\in(-1,1)$. Theorem \ref{Thm:renorm} shows that the operator $P-P^0_-$ is in general \emph{not} trace-class: if $P-P^0_-\in\gS_1(\gH_\Lambda)$, then it must hold $\tr_{P^0_-}(P-P^0_-)=\tr(P-P^0_-)=\int_{\R^3}\rho_{P-P^0_-}$.

In our model we have two possible definitions of the charge of the system: $\int_{\R^3}\nu-\tr_{P^0_-}(P-P^0_-)$ and $\int_{\R^3}(\nu-\rho_{P-P^0_-})$. In practice it is the electrostatic field induced by the nucleus (together with the vacuum polarization density) which is measured, hence it is more natural to define the charge by means of the density. By \eqref{relation_charge}, the total Coulomb potential is, at infinity,
$$\alpha(\nu-\rho_{P-P^0_-})\ast\frac1{|x|}\;\underset{|x|\to\ii}{\sim}\; \alpha\frac{\int_{\R^3}(\nu-\rho_{P-P^0_-})}{|x|}=\frac{\frac{\alpha}{1+\alpha B_\Lambda}\left(\int_{\R^3}\nu-\tr_{P^0_-}(P-P^0_-)\right)}{|x|}.$$
Let us assume for simplicity that we put in the vacuum ($\mu=0$) a nucleus containing $\int_{\R^3}\nu=Z$ protons and which is small enough in the sense that $\norm{\nu}_\cC\ll1$. Then $\tr_{P^0_-}(P-P^0_-)=0$ by \cite[Theorem 3]{HaiLewSer-05b} and we see that at infinity the potential induced by the nucleus is not $\alpha Z/|x|$ as expected, but rather $\alphaph Z/|x|$ where
\begin{equation}
\boxed{\alpha_{\rm ph}=Z_3 \alpha \;,\quad \text{ with }\quad Z_3=\frac1{1+\alpha B_\Lambda}=1-\alpha_{\rm ph} B_\Lambda.}
\label{charge_renormalization}
\end{equation}
The charge renormalization constant $Z_3$ is well known in QED\footnote{The renormalization constant $Z_3$ should not be confused with the nuclear charge $Z=\int_ {\R^3}\nu$.} \cite{Dyson-49b,BjorkenDrell-65,ItzyksonZuber}. The value of $\alpha$ is not observable, $\alpha_{\rm ph}$ is the real physical constant since we always observe the nucleus together with the vacuum polarization density. Its experimental value is $\alpha_{\rm ph}\simeq 1/137\,.$

In our theory we must fix $\alpha_{\rm ph}$ and not $\alpha$. Using \eqref{charge_renormalization} we can express any physical quantity in terms of $\alphaph$ and $\Lambda$ only. Unfortunately it holds $\alpha_{\rm ph} B_\Lambda<1$ hence it makes no sense to take $\Lambda\to\infty$ while keeping $\alpha_{\rm ph}$ fixed (this is the so-called Landau pole \cite{Landau-55}) and one has to look for a weaker definition of renormalizability.
The cut-off $\Lambda$ which was first introduced as a mathematical trick to regularize the model has actually a physical meaning. Because of the above constraint $\alpha_{\rm ph} B_\Lambda<1$, a natural scale occurs beyond which the model does not make sense. Fortunately, this scale is of the order $e^{3\pi/2\alphaph}$, a huge number for $\alphaph\simeq1/137$.

It is more convenient to change variables and take as new parameters $\alphaph$ and $Z_3=1-\alphaph B_\Lambda$, with the additional constraint that $0<Z_3<1$. The new parameter $Z_3$ is now independent of $\alphaph$ and the natural question arises whether predicted physical quantities will depend very much on the chosen value of $0<Z_3<1$. The purpose of this paper is to prove that the asymptotics of any physical quantity in the regime $\alphaph\ll1$ is actually \emph{independent of $Z_3$} to any order in $\alphaph$, which is what we call \emph{asymptotic renormalizability}. Note that fixing $Z_3\in(0,1)$ amounts to take  $\Lambda\simeq Ce^{3\pi(1-Z_3)/2\alphaph}\gg1$.

Instead of looking at all possible physical observables, it is convenient to define a renormalized density $\rho_{\rm ph}$. Following \cite{HaiLewSer-05b}, we define it by the relation
\begin{equation}
\alpha_{\rm ph}\rho_{\rm ph}=\alpha\big(\nu-\rho_{P-P^0_-}\big)
\label{eq:def_rho_ph}
\end{equation}
in such a way that $D=D^0-\alpha_{\rm ph}\rho_{\rm ph}\ast|x|^{-1}$. This procedure is similar to wavefunction renormalization. By uniqueness of $\rho_{P-P^0_-}$ we can see $\rho_{\rm ph}$ as a function of $\alpha_{\rm ph}$, $\nu$, $\mu$ and $\Lambda$ (or $Z_3$). For the sake of clarity we will not emphasize the dependence in $\nu$ and $\mu$ which will be fixed quantities. Also we will use the same notation $\rhoph(\alphaph,\Lambda)$ or $\rhoph(\alphaph,Z_3)$, depending on the context. The self-consistent equation for $\rhoph$ was derived in \cite{HaiLewSer-05b} and it is mentioned below in Section \ref{sec:def_nu_k}.

From now on, we will assume that 
$$\boxed{\mu=0.}$$
For small external densities $\nu$, this means that we will be looking at the vacuum polarization in the presence of the nucleus, without considering any real electron (that is, $\rhoph$ is the renormalized density of the nucleus containing both the bare density $\nu$ and the vacuum polarization density $\rho_{P-P^0_-}$). We will explain in Section \ref{sec:def_nu_k} that one can expand $\rho_{\rm ph}=\rho_{\rm ph}(\alpha_{\rm ph},\Lambda)$ as follows:
\begin{equation}
\rho_{\rm ph}(\alpha_{\rm ph},\Lambda)=\sum_{n=0}^\infty (\alpha_{\rm ph})^n\nu_{n,\Lambda}
\label{series_cut_off}
\end{equation}
where $\{\nu_{n,\Lambda}\}_n\subset L^2(\R^3)\cap\cC$ is a sequence depending only on the external density $\nu$ and the cut-off $\Lambda$. This sequence is defined below in Section \ref{sec:def_nu_k}. The series \eqref{series_cut_off} has a positive  radius of convergence, which is however believed to shrink to zero when $\Lambda\to\ii$.

Assuming $\widehat\nu$ decays fast enough (see condition \eqref{condition_nu}), we will prove that for any fixed $n$, the limit  $\nu_{n,\Lambda}\to\nu_{n}$ exists in $L^2(\R^3)\cap\mathcal{C}$. This is what is usually meant by renormalizability in QED: each term of the perturbation series in powers of the physical $\alpha_{\rm ph}$ has a limit when the cut-off is removed. The sequence $\{\nu_n\}_n$ is the one which is calculated in practice \cite{BjorkenDrell-65,GreMulRaf-85,EngDre-87,Engel-02}. One has for instance $\nu_0=\nu$ and 
$$\nu_1\ast|x|^{-1}=\frac{1}{3\pi}\int_1^\infty dt\, (t^2-1)^{1/2}\left[\frac2{t^2}+\frac{1}{t^4}\right]\int_{\R^3} e^{-2|x-y|t}\frac{\nu(y)}{|x-y|}\, dy,$$
the \emph{Uehling potential} \cite{Uehling-35,Serber-35}. All the others $\nu_n$ can be calculated by induction in terms of $\nu_0,...,\nu_{n-1}$, as is explained below in Section \ref{sec:def_nu_k}.

The next natural question is to understand the link between the well-defined, cut-off dependent, series \eqref{series_cut_off} and the \emph{formal series} $\sum_{n=0}^\infty(\alpha_{\rm ph})^n\nu_{n}$. Recall that $\alpha_{\rm ph}B_\Lambda<1$ by construction, so it is in principle not allowed to take the limit $\Lambda\to\infty$ while keeping $\alpha_{\rm ph}$ fixed: we rather want to think of $Z_3=1-\alphaph B_\Lambda$ as being fixed. The main result in this paper is the following

\begin{thm}[Asymptotic renormalization of the nuclear charge density]
\label{Thm:asymptotics}
Consider a function $\nu \in L^2(\R^3) \cap \cC$ such that
\begin{equation}
\label{condition_nu}
\int_{\R^3} \log (1 + |k|)^{2 N + 2} |\widehat\nu(k)|^2 dk < \ii
\end{equation}
for some integer $N$. Let $\rhoph(\alphaph,Z_3)$ be the unique physical density defined by \eqref{eq:def_rho_ph} with $\mu = 0$, $Z_3=1-\alphaph B_\Lambda$ and $\alphaph= Z_3\alpha$.

Then, for every $0<\epsilon<1$, there exist two constants $C(N, \epsilon, \nu)$ and $a(N, \epsilon, \nu)$, depending only on $N$, $\epsilon$ and $\nu$, such that one has
\begin{equation}
\label{eq:asymptotics}
\boxed{\norm{\rhoph(\alphaph, Z_3) - \sum_{n = 0}^N \nu_n (\alphaph)^n}_{L^2(\R^3) \cap \cC} \leq C(N, \epsilon, \nu)\; \alphaph^{N + 1}}
\end{equation}
for all $0 \leq \alphaph \leq a(N, \epsilon, \nu)$ and all $\epsilon \leq Z_3 \leq 1 - \epsilon$.
\end{thm}

The interpretation of Theorem \ref{Thm:asymptotics} is that the renormalized density $\rhoph(\alphaph,Z_3)$ is \emph{asymptotically} (meaning up to any fixed order $N$) given by the formal series $\sum_{n\geq0}(\alphaph)^n\nu_n$, \emph{uniformly in the renormalization parameter $Z_3$} in the range $\epsilon\leq Z_3\leq 1-\epsilon$. Therefore, for a very large range of cut-offs, essentially 
$$C_1e^{{3\epsilon\pi}/{2\alpha_{\rm ph}}}\leq \Lambda\leq C_2e^{{3(1-\epsilon)\pi}/{2\alpha_{\rm ph}}}$$
the result is independent of $\Lambda$ and it is given by the formal series $\sum_{n\geq0}(\alphaph)^n\nu_n$. Our formulation of renormalizability is more precise than the requirement that each $\nu_{n,\Lambda}$ converges. It also leads to the formal perturbation series in a very natural way.

A natural question is to ask for the convergence of the perturbation series $\sum_{n\geq0}(\alphaph)^n\nu_n$. It was argued by Dyson in \cite{Dyson-52} that it is probably \emph{divergent}, but we are unable to transform his argument into a rigorous mathematical proof. We will make more comments on the series $\sum_{n\geq0}(\alphaph)^n\nu_n$ at the end of next section.

\begin{remark}
We will provide explicit formulas for the sequence $\{\nu_n\}$ later in Section \ref{sec:def_nu_k}.  We will in particular see in the proof that under Assumption \eqref{condition_nu}, one has $\nu_n\in L^2(\R^3)\cap\cC$ for all $0\leq n\leq N$. Therefore the approximation series of order $N$ appearing in \eqref{eq:asymptotics}, $\sum_{n=0}^N(\alphaph)^n\nu_n$, is a well-defined function of $L^2(\R^3)\cap\cC$.
\end{remark}

\begin{remark}
The space $L^2(\R^3)\cap\cC$ is the natural space which occurs in this theory. In particular the Coulomb norm is nothing but the classical electrostatic energy which appears in the reduced BDF energy functional. Our result can be extended to Sobolev spaces $H^s(\R^3)$ provided $\nu$ is smooth enough.
\end{remark}

\begin{remark}
It would be interesting to extend this result to the case of atoms with 'real' electrons. This amounts to taking $\mu$ sufficiently close to $1$ at the same time as $\alphaph$ is small. However this case is more difficult than what is done here: an additional expansion of the electronic charge density in powers of $\alphaph$ is needed.
\end{remark}

The proof of Theorem \ref{Thm:asymptotics} (given in Section \ref{sec:proof_thm_asymptotics} below) is divided into two steps. We first estimate the difference (see Lemma \ref{Lem:reste}) 
\begin{equation}
\norm{\rhoph(\alphaph,\Lambda)-\sum_{n = 0}^N \nu_{n,\Lambda} (\alphaph)^n}_{L^2(\R^3)\cap\cC}\leq C_1(N,\epsilon,\nu)(\alphaph)^{N+1}
\label{estim_rest_intro} 
\end{equation}
for a constant $C_1(N,\epsilon,\nu)$ depending only on $N$, $\epsilon$ and $\nu$, and under the assumption that $\epsilon\leq Z_3=1-\alphaph B_\Lambda\leq 1-\epsilon$. This amounts to expanding the solution of the self-consistent equation \eqref{eq:SCF_cut_off} up to the $N$th order in $\alphaph$ while controlling the error term uniformly in $\Lambda$. Then we show in Lemma \ref{Lem:diff_nu} that
\begin{equation}
\forall 0\leq n\leq N,\qquad \norm{\nu_{n,\Lambda}-\nu_n}_{L^2(\R^3)\cap\cC}\leq \frac{C_2(N,\nu)}{(B_\Lambda)^{N+1-n}}
\label{estim_nu_nnu_n_Lambda_intro} 
\end{equation}
for a constant $C_2(N,\nu)$ depending only on $N$ and $\nu$, leading to the bound
\begin{equation}
\norm{\sum_{n = 0}^N \nu_{n,\Lambda} (\alphaph)^n-\sum_{n = 0}^N \nu_{n} (\alphaph)^n}_{L^2(\R^3)\cap \cC}\leq (\alphaph)^{N+1}\frac{C_2(N,\nu)(1-\epsilon^{N+1})}{\epsilon^{N+1}(1-\epsilon)}
\label{estim_series_Lambda_no} 
\end{equation}
since by assumption $(B_\Lambda)^{-N-1+n}\leq (\alphaph/\epsilon)^{N+1-n}$. The main result then follows from \eqref{estim_rest_intro} and \eqref{estim_series_Lambda_no}. All these bounds strongly use the explicit recursion relations defining the sequences $\{\nu_{n,\Lambda}\}$ and $\{\nu_n\}$, as well as tedious estimates on the nonlinear terms appearing in these relations.

The rest of the paper is organized as follows. In Section \ref{sec:def_nu_k} we define the sequences $\{\nu_{n,\Lambda}\}$ and $\{\nu_{n}\}$ by their respective recursion formulas and we discuss some properties of the latter. In particular, in Theorem \ref{Thm:prop_nu_n}, we give a simple estimate on $\norm{\nu_n}_{L^2(\R^3)\cap\cC}$. In Section \ref{sec:U_F} we present estimates on the different terms appearing in the recursion formulas. Of particular interest will be the density $\nu_{1,\Lambda}$ giving rise to the Uehling potential. In Fourier space, we have
$\widehat{\nu}_{1,\Lambda}(k)=U_\Lambda(k)\widehat{\nu}(k)$
for an explicit function $U_\Lambda(k)$ which is studied in Section \ref{sec:U}. The proofs of Theorems \ref{Thm:prop_nu_n} and \ref{Thm:asymptotics} are respectively provided in Sections \ref{sec:proof_thm_prop_nu_n} and  \ref{sec:proof_thm_asymptotics}. Some other technical proofs are provided in Appendices \ref{sec:appendix}, \ref{sec:proof_Prop:ULU} and \ref{sec:proof_estim_F}.

\bigskip

\noindent\textbf{Acknowledgment.} The authors are grateful to Christian Brouder for interesting comments. M.L. would like to thank Jan Derezi\'nski and Jan Philip Solovej for stimulating discussions.

\section{The two sequences $\{\nu_{n,\Lambda}\}$ and $\{\nu_{n}\}$}
\label{sec:def_nu_k}

In this section we derive formulas for $\{\nu_{n,\Lambda}\}$ and $\{\nu_{n}\}$, and we make some comments on the latter.

\subsection{Definition of $\{\nu_{n,\Lambda}\}$ and $\{\nu_n\}$}

We start with the self-consistent equation \eqref{eq:SCF_cut_off} with cut-off, assuming $\mu=0$. 
Note that in the regime of interest in Theorem \ref{Thm:asymptotics},
we have $\alpha=\alphaph/{Z_3}\leq {\alphaph}/{\epsilon}$.
When $\alpha \pi^{1/6}2^{11/6}\norm{\nu}_\cC<1$, it is known that $0\notin\sigma(D)$ hence $\delta=0$ in \eqref{eq:SCF_cut_off}, see \cite[Theorem 3]{HaiLewSer-05b} and \cite[Lemma 11]{GraLewSer-09}. Therefore assuming $a(\nu,N,\epsilon)\leq \epsilon(\pi^{1/6}2^{11/6}\norm{\nu}_\cC)^{-1}$ in Theorem \ref{Thm:asymptotics}, we automatically have that $\delta=0$ and $P=P^2$ is unique.

The idea is then to expand the self-consistent equation
\begin{equation}
P=\chi_{(-\ii,0)}\Pi_\Lambda\big(D^0+\alpha(\rho_{P-P^0_-}-\nu)\ast|x|^{-1}\big)\Pi_\Lambda
\label{eq:SCF_P_only} 
\end{equation}
in powers of $\alpha$ by means of the resolvent formula. This method was already used in \cite{HaiLewSer-05a} to prove existence and uniqueness of solutions. We define
\begin{equation}
\label{def:FnL}
F_{n,\Lambda}(\mu_1,...,\mu_k):=\rho\left[ \frac{1}{2\pi}\int_{-\ii}^\ii\frac{1}{D^0+i\eta}\prod_{j=1}^n\left(\Pi_\Lambda\;\mu_j\ast\frac{1}{|x|}\;\Pi_\Lambda\;\frac{1}{D^0+i\eta}\right)d\eta\right]
\end{equation}
where we recall that $\Pi_\Lambda$ is the orthogonal projector onto $\gH_\Lambda$ in $L^2(\R^3,\C^4)$ and $\mu_1,...,\mu_n\in \cC$. We will always use the simplified notation $F_{n,\Lambda}(\mu):=F_{n,\Lambda}(\mu,...,\mu)$ and $\nu_\Lambda:=\cF^{-1}(\widehat{\nu}\1_{B(0,2\Lambda)})$. Note that by Furry's theorem $F_{2j,\Lambda}\equiv0$ for all $j$, see \cite[p. 547]{HaiLewSer-05a}. We also introduce
$$F_\Lambda(\mu):=\sum_{n\geq3}F_{n,\Lambda}(\mu).$$

The self-consistent equation \eqref{eq:SCF_P_only} may then be written in terms of the density in Fourier space \cite{HaiLewSer-05a,HaiLewSer-05b}, as 
\begin{equation}
\widehat{\rho_{P-P^0_-}}(k)=-\alpha B_\Lambda(k)(\widehat{\rho_{P-P^0_-}}(k)-\widehat{\nu}_\Lambda(k))+\widehat{F}_{\Lambda}\big(\alpha(\nu-\rho_{P-P^0_-})\big)
\label{eq:SCF_rho}
\end{equation}
where the function $B_\Lambda(k)$ is given by
\begin{multline}
B_\Lambda(k)  = \frac1\pi\int_{0}^{Z_\Lambda(|k|)}\frac{z^2-z^4/3}{(1-z^2)(1+|k|^2(1-z^2)/4)}dz\\
+ \frac{|k|}{2 \pi} \int_0^{Z_\Lambda(|k|)} \frac{z - z^3/3}{\sqrt{1 + \Lambda^2} - |k| z/2} dz
\label{expression_B}
\end{multline}
with
$Z_\Lambda(r)=\big(\sqrt{1+\Lambda^2}-\sqrt{1+(\Lambda-r)^2}\big)/r$, see \cite{GraLewSer-09}.
The formula for $B_\Lambda(k)$ is well-known (but in most previous works the second term was ignored, see for instance \cite{PauRos-36}).

Defining $U_\Lambda(|k|)=B_\Lambda-B_\Lambda(k)$ where $B_\Lambda=B_\Lambda(0)$ and $0\leq U_\Lambda(|k|)\leq B_\Lambda$ with $U_\Lambda(2\Lambda)=B_\Lambda$, we get the renormalized equation 
\begin{equation}
{\big(1-\alphaph U_\Lambda\big)\widehat\rho_{\rm ph}+\widehat{F}_\Lambda(\alphaph\rho_{\rm ph})=\widehat{\nu_\Lambda}}
\label{eq:SCF_recast}
\end{equation}
with the renormalized coupling constant $\alphaph:={\alpha}/{(1+\alpha B_\Lambda)}$ and the renormalized density $\alphaph\rho_{\rm ph}=\alpha(\nu-\rho_Q)$ (see \cite{HaiLewSer-05b}).
For convenience, we will denote by $\cU_\Lambda$ the operator of multiplication by the function $U_\Lambda(|k|)$ in the Fourier domain. Hence we can write the self-consistent equation \eqref{eq:SCF_recast} in direct space as
\begin{equation}
\boxed{\big(1-\alphaph \cU_\Lambda\big)\rho_{\rm ph}+{F}_\Lambda(\alphaph\rho_{\rm ph})={\nu_\Lambda}.}
\label{eq:SCF_recast2}
\end{equation}

We now expand the unique solution $\rho_{\rm ph}=\rho_{\rm ph}(\alphaph,\Lambda)$ of \eqref{eq:SCF_recast2} in powers of $\alphaph$. Writing a formal series
\begin{equation}
\rho_{\rm ph}=\sum_{n\geq0}(\alphaph)^n\nu_{n,\Lambda}
\label{eq:series_rhoph2}
\end{equation}
we find that the functions $\nu_{n,\Lambda}$ must satisfy the following recurrence relation
\begin{equation}
 \left\{\begin{array}{ll}
\nu_{0,\Lambda}=\nu_\Lambda,&\\
 & \\
\nu_{1,\Lambda}=\cU_\Lambda\nu_\Lambda,&\\
& \\
\displaystyle\nu_{n,\Lambda}=\cU_\Lambda{\nu}_{n-1,\Lambda}+\sum_{j=3}^n\;\sum_{n_1+\cdots+ n_j=n-j}{F}_{j,\Lambda}\big(\nu_{n_1,\Lambda},...,\nu_{n_j,\Lambda}\big),& \forall n\geq2.
\end{array}
\right.
\label{def:nu_nL}
\end{equation}
Note that the operator $\cU_\Lambda$ is bounded by $U(2\Lambda)=B_\Lambda$ on $\gH_\Lambda$ and that, as we will see later in Corollary \ref{Cor:estimF},  each $F_{j,\Lambda}$ is continuous on $\cC^j$ with values in $L^2\cap\cC$. The sequence $\{\nu_{n,\Lambda}\}$ is thus well-defined in $L^2\cap\cC$. Using estimates from \cite{HaiLewSer-05a} it can be proven that the series \eqref{eq:series_rhoph2} has a finite radius of convergence in $L^2\cap\cC$, but this is not needed for the moment and we can stay at a formal level in this section.

We can now formally pass to the limit as $\Lambda\to\ii$ and define by induction a sequence $\{\nu_n\}$ by
\begin{equation}
 \left\{\begin{array}{ll}
\nu_{0}=\nu,&\\
 & \\
\nu_{1}=\cU\nu,&\\
& \\
\displaystyle\nu_{n}=\cU{\nu}_{n-1,}+\sum_{j=3}^n\;\sum_{n_1+\cdots+ n_j=n-j}{F}_{j}\big(\nu_{n_1,},...,\nu_{n_j,}\big),& \forall n\geq2.
\end{array}
\right.
\label{def:nu_n}
\end{equation}
where the $F_j$ are defined similarly as the $F_{j,\Lambda}$ with $\Pi_\Lambda$ removed and $\cU$ is the operator of multiplication by the function $U(|k|)$ in the Fourier domain, defined by
\begin{align}
\label{def:U}
U(r):=\lim_{\Lambda\to\ii}U_\Lambda(r)&=\frac{r^2}{4\pi}\int_0^1\frac{z^2-z^4/3}{1+\frac{r^2(1-z^2)}{4}}\,dz\\
&\nonumber=\frac{12-5r^2}{9\pi r^2}+\frac{\sqrt{4+r^2}}{3\pi r^3}(r^2-2)\log\left(\frac{\sqrt{4+r^2}+r}{\sqrt{4+r^2}-r}\right).
\end{align}

\subsection{On the series $\{\nu_n\}$}

The recursion formula \eqref{def:nu_n} defining $\{\nu_n\}$ contains two terms. The first term $\cU\nu_{n-1}$ is a simple multiplication operator in Fourier space, by the function $U(|k|)$ which diverges at infinity. The second term involves the nonlinear functions $F_j$'s. If only the first term with $\cU$ were present, the series $\sum_{n\geq0}\nu_n(\alphaph)^n$ would only converge when the Fourier transform $\widehat\nu$ has a compact support, the radius of convergence depending on the size of this support. If only the nonlinear terms were present, the series would have a finite radius of convergence by the estimates of \cite{HaiLewSer-05a} and of Section \ref{sec:F}.

However when the two terms are combined, the situation is much more complicated. The nonlinear terms act like convolutions in Fourier space, hence even if $\widehat\nu$ has a compact support in the Fourier domain, the support of $\widehat\nu_n$ will probably grow with $n$. A careful study of the mixed effect of the multiplication by the divergent function $U$ and the nonlinearities seems rather difficult.
We will prove the following estimate:

\begin{thm}[Estimate on $\{\nu_n\}_{n \geq1}$]
\label{Thm:prop_nu_n}
There exist universal constants $A$ and $K$ such that
\begin{multline}
\label{estim:nu_n}
\norm{(1 + \cU)^{m - n} \nu_n}_{L^2 \cap \cC}\\
\leq A^{n + 1} \max \bigg\{ \| (1 +  \cU)^m \nu \|_{L^2 \cap \cC}\;,\; \big(K\log(m)\big)^\frac{m n}{2} \| (1 + \cU)^m \nu \|_{L^2 \cap \cC}^{n + 1} \bigg\},
\end{multline}
for all $\Lambda \geq 1$, $m \in \N$ and $0 \leq n \leq m$.
\end{thm}

Even if we assume that $\nu$ decays fast enough in Fourier space, for instance
$$\forall n\geq0,\qquad \norm{(1+\cU)^{n}\nu}_{L^2\cap\cC}\leq C^n ,$$
the above estimate \eqref{estim:nu_n} does \emph{not} imply that the series $\sum_{n\geq0}\nu_n(\alphaph)^n$ is convergent for $\alphaph$ small enough. Although our estimate \eqref{estim:nu_n} is certainly far from optimal, as we have already mentionned, it is expected that the series does not converge in any appropriate sense \cite{Dyson-52}.

It is sometimes argued that the series could be \emph{Borel summable}. The Borel transform is defined by
$$\cB(t)=\sum_{n\geq0}\frac{t^n}{n!}\nu_n.$$
If $\cB(t)$ is a convergent series (for an appropriate norm) having a holomorphic extension to a domain containing the positive real line, such that
$$\tilde{\cB}(\alphaph):=\int_0^\ii \cB(t)e^{-\frac{t}{\alphaph}}dt$$
makes sense in an appropriate neighborhood of $\alphaph=0$, one may see $\tilde{\cB}(\alphaph)$ as the physical density, whose series $\sum_{n\geq0}(\alphaph)^n\nu_n$ is only asymptotic. Proving such results mathematically is hard, even for the model studied in this paper. Our estimate \eqref{estim:nu_n} does not even allow to define the Borel transform $\cB(t)$ in $L^2(\R^3)\cap\cC$.

But Borel summability is not the only tool to construct a physical density providing the correct asymptotic series. For the model studied in the present paper, we have several natural families of functions of $\alphaph$, the cut-off densities 
\begin{equation}
\rhoph\big(\alphaph,Ce^{3(1-Z_3)\pi / 2\alphaph}\big)
\label{eq:rho_ph_Lambda} 
\end{equation}
obtained by minimizing the reduced BDF energy with a cut-off $\Lambda=Ce^{3(1-Z_3)\pi / 2\alphaph}$ and using the relation \eqref{eq:def_rho_ph}. Each such density \eqref{eq:rho_ph_Lambda} has (for fixed $C$ and $0<Z_3<1$) the required asymptotic series in $\alphaph$ by Theorem \ref{Thm:asymptotics}, and it solves the self-consistent equation \eqref{eq:SCF_cut_off} with the corresponding cut-off $\Lambda$. Furthermore this solution has the benefit of being well-defined even when $\alphaph$ is not small, allowing for the description of nonperturbative physical events. 

The rest of the paper is devoted to the proofs of Theorems \ref{Thm:asymptotics} and \ref{Thm:prop_nu_n}.

\section{Some preliminary results}
\label{sec:U_F}

In this section we state two preliminary results that will be useful in the proof of our main results, Theorems \ref{Thm:asymptotics} and \ref{Thm:prop_nu_n}. The corresponding lengthy calculations will be provided later in Appendices \ref{sec:appendix}, \ref{sec:proof_Prop:ULU} and \ref{sec:proof_estim_F}.

\subsubsection*{Notation} In the whole paper we use the notation $E(r)=(1+\vert r\vert^2)^{1/2}$ for $r\in\R^3$ or $r\in\R$.

\subsection{The Uehling multiplier $U$}
\label{sec:U}

The operator $\cU$, defined previously as the multiplication by the function $U$ in the Fourier domain, plays a major role in the definition of the sequence $\{\nu_n\}$.
In this section, we provide precise estimates quantifying the convergence of $U_\Lambda$ towards $U$ when $\Lambda \to \ii$, 
which will be very useful in the proof of Theorem \ref{Thm:asymptotics}. 

\begin{prop}
\label{Prop:ULU}
Let $\Lambda \geq 1$ and denote by $U_\Lambda$, the function defined on $\R_+$ by
\begin{equation}
\label{def:UL}
U_\Lambda(r) = \left\{\begin{array}{ll}
B_\Lambda - B_\Lambda(r)&\text{when $0\leq r\leq 2\Lambda$,}\\
0 & \text{otherwise.}
\end{array}\right.
\end{equation}
Then,  for all $r \in \R_+$ it holds $\lim_{\Lambda\to\ii}U_\Lambda(r) =U(r)$.
Moreover, for $\kappa_0 = {15 \pi}/{2}$,
\begin{equation}
\label{estim:ULU}
\forall m\geq0,\qquad\norm{\frac{U_\Lambda - U}{(1 + U)^{m + 1}}}_{L^\infty} \leq \kappa_0^{m + 3} \max \Big\{ \frac{1}{(1 + B_\Lambda)^m}, \frac{1}{E(2 \Lambda)} \Big\}.
\end{equation}
Finally, one has for a universal constant $\kappa_1$ (given in Lemma \ref{Lem:ULU} below)
\begin{equation}
\label{elise}
\forall 0 \leq r \leq 2 \Lambda,\qquad 0 \leq U_\Lambda(r) \leq \kappa_1 \big( 1 + U(r) \big).
\end{equation}
\end{prop}

Proposition \ref{Prop:ULU} is proved in Appendix \ref{sec:proof_Prop:ULU}. Note that the uniform estimate \eqref{estim:ULU} will later yield our estimate \eqref{estim_nu_nnu_n_Lambda_intro} on $\nu_{n,\Lambda}-\nu_n$ (see Lemma \ref{Lem:diff_nu}). More properties of $U$ and $U_\Lambda$ are provided later in Appendix \ref{sec:appendix}.

\subsection{The nonlinear terms $F_{n, \Lambda}$ and $F_{n}$}
\label{sec:F}

In this section, we provide estimates on the functions $F_{n, \Lambda}$ and $F_{n}$, which will be one main ingredient in the proof of Theorem \ref{Thm:asymptotics}. We recall that $F_{2 n, \Lambda} = F_{2 n} = 0$ by Furry's theorem (see \cite[p. 547]{HaiLewSer-05a}). In order to state our main result, we introduce the functions
\begin{equation}
\label{def:FnLe}
F_{n, \Lambda}^{\epsilon}(\mu) = \rho \bigg[ \frac{1}{2 \pi} \int_{- \infty}^{\ii} \frac{1}{D^0 + i \eta} \prod_{j = 1}^n \Big( \Pi_\Lambda^{(\epsilon_j)} \mu_j \ast \frac{1}{|x|} \Pi_\Lambda^{(\epsilon_{j + 1})} \frac{1}{D^0 + i \eta} \Big) d\eta \bigg],
\end{equation}
for any $n \geq 3$, $\mu = (\mu_1, \cdots, \mu_n) \in \cC^n$ and $\epsilon = (\epsilon_1, \cdots, \epsilon_{n + 1}) \in \{- 1, 0, 1 \}^{n + 1}$. Here, we have used the notation
\begin{equation}
\label{def:PiL}
\Pi_\Lambda^{(1)} := \Pi_\Lambda, \quad \Pi_\Lambda^{(- 1)} := 1 - \Pi_\Lambda \quad {\rm and} \quad \Pi_\Lambda^{(0)} := 1 = \Pi_\Lambda^{(1)} +\Pi_\Lambda^{(- 1)}.
\end{equation}
The main result of this section is the following

\begin{prop}[Estimates on $F_{n, \Lambda}^\epsilon$]
\label{Prop:estimF}
Let $m \in \N$, $\Lambda \geq 1$ and $\epsilon \in \{ - 1, 0, 1 \}^{n + 1}$. Assume that $n \geq 3$. Then, there exist universal constants $C$ and $K$ such that
\begin{equation}
\label{estim:F}
\norm{(1 + \cU)^m F_{n, \Lambda}^\epsilon(\mu)}_{L^2 \cap \cC} \leq \frac{C^n (K \log n)^m}{\Lambda^{n(\epsilon)/{24}}} \prod_{j = 1}^n \norm{(1 + \cU)^m \mu_j}_{\cC},
\end{equation}
for all $\mu = (\mu_1, \cdots, \mu_n) \in \cC^n$. Here, $n(\epsilon) = 1$, if at least one $\epsilon_j$ is equal to $- 1$, and $n(\epsilon) = 0$ otherwise.
\end{prop}

By \eqref{def:FnL}, \eqref{def:FnLe} and \eqref{def:PiL}, we can write $F^{(1, \cdots, 1)}_{n, \Lambda} = F_{n, \Lambda}$ and $F^{(0, \cdots, 0)}_{n, \Lambda} = F_{n}$. Therefore the following is a byproduct of \eqref{estim:F}:

\begin{corollary}
\label{Cor:estimF}
Let $m \geq0$, $\Lambda \geq 1$ and $n \geq3$ an odd integer. Then,
\begin{multline}
\label{estim:FnL}
\max \Big\{ \norm{(1 + \cU)^m F_{n, \Lambda}(\mu)}_{L^2 \cap \cC}, \norm{(1 + \cU)^m F_{n}(\mu)}_{L^2 \cap \cC} \Big\}\\
\leq C^{n} (K \log n )^m \prod_{j = 1}^{n} \norm{(1 + \cU)^m \mu_j}_{\cC},
\end{multline}
for any $\mu = (\mu_1, \cdots, \mu_{n}) \in \cC^{n}$. Here, $C$ and $K$ refer to the universal constants given by Proposition \ref{Prop:estimF}. In particular, the functions $F_{n, \Lambda}$ and $F_{n}$ are continuous on $\cC^{n}$ with values in $L^2 \cap \cC$.
\end{corollary}

Recall $F_{2k}=F_{2k,\Lambda}=0$ hence only the case of $n$ being an odd integer is relevant. The estimates of Proposition \ref{Prop:estimF} are an adaptation of ideas of \cite{HaiLewSer-05a}, in which similar bounds were computed (see, e.g.,  Lemmas 15 and 16 in \cite{HaiLewSer-05a}). Notice however that the projector $\Pi_\Lambda$ was never mentionned in \cite{HaiLewSer-05a} since $\Lambda$ was a fixed number. We focus here on the limit $\Lambda \to \ii$ and we need to quantify the dependence on $\Lambda$ of the estimates on the functions $F_{n, \Lambda}^\epsilon$.
The proof of Proposition \ref{Prop:estimF} is provided below in Appendix \ref{sec:proof_estim_F}. The factor $(K\log n)^m$ comes from \eqref{U:prod} of Lemma \ref{Lem:U-prod} and the constant $K$ is also the one appearing in Theorem \ref{Thm:prop_nu_n}.

\section{Proof of Theorem \ref{Thm:prop_nu_n}}
\label{sec:proof_thm_prop_nu_n}

This section is devoted to the proof of our estimate \eqref{estim:nu_n} on the $n$th order density $\nu_n$. The definition of $\nu_{n, \Lambda}$ being very similar to that of $\nu_n$, our proof also provides the following

\begin{prop}[Estimates on $\nu_{n, \Lambda}$]
\label{Prop:nu_nL}
There exists $A>0$ such that
\begin{multline}
\label{estim:nu_nL}
 \norm{(1 + \cU)^{m - n} \nu_{n, \Lambda}}_{L^2 \cap \cC}\\
\leq A^{n + 1} \max \left\{ \| (1 + \cU)^m \nu \|_{L^2 \cap \cC}, \big(K\log(m)\big)^\frac{m n}{2} \| (1 + \cU)^m \nu \|_{L^2 \cap \cC}^{n + 1} \right\},
\end{multline}
for any $\Lambda \geq 1$, $m \in \N$ and $0 \leq n \leq m$.
\end{prop}

We postpone the proof of Proposition \ref{Prop:nu_nL} and first complete that of Theorem \ref{Thm:prop_nu_n}.

\begin{proof}[Proof of Theorem \ref{Thm:prop_nu_n}]
We split the proof into three steps. First, we estimate by means of \eqref{estim:FnL}, the following norms:
$J_{m, n} := \| (1 + \cU)^{m - n} \nu_n \|_{L^2 \cap \cC}$.

\setcounter{step}{0}
\begin{step}
\label{P1}
Let $m \in \N$ and denote
$$P_m(t) := \sum_{n = 0}^m J_{m, n} t^n.$$
The polynomial $P_m(t)$ satisfies for any $t \geq 0$
\begin{equation}
\label{haiti}
P_m(t) \leq \big( 1 + t + t^2) \| (1 + \cU)^m \nu \|_{L^2 \cap \cC} + \cQ_m(t P_m(t)),
\end{equation}
where ($C$ and $K$ are the constants of Proposition \ref{Prop:estimF})
\begin{equation}
\label{def:Q_m}
\cQ_m(u) := u + \sum_{j = 3}^m C^j (K \log j )^{m - j} u^j.
\end{equation}
\end{step}

Let us assume first that $n=0,1,2$. By \eqref{def:nu_n}, we then have $\nu_n= \cU^n \nu$, hence
\begin{equation}
\label{port}
\forall n=0,1,2,\qquad J_{m, n} = \| (1 + \cU)^{m - n} \cU^n \nu \|_{L^2 \cap \cC} \leq \| (1 + \cU)^m \nu \|_{L^2 \cap \cC}.
\end{equation}
We now turn to the case $n \geq 3$. By \eqref{def:nu_n}, we have
\begin{multline*}
J_{m, n} \leq  \| (1 + \cU)^{m - n} \cU \nu_{n - 1} \|_{L^2 \cap \cC}\\
+  \sum_{3 \leq 2 j + 1 \leq n} \  \sum_{\underset{k = 1}{\overset{2 j + 1}{\sum}} n_k = n - 2 j - 1} \| (1 + \cU)^{m - n} F_{2 j + 1} \big(\nu_{n_1}, \cdots, \nu_{n_{2 j + 1}} \big) \|_{L^2 \cap \cC},
\end{multline*}
hence, by Corollary \ref{Cor:estimF},
\begin{align*}
J_{m, n} \leq \| (1 + \cU)^{m - n + 1} & \nu_{n - 1} \|_{L^2 \cap \cC} + \sum_{3 \leq 2 j + 1 \leq n} \ \  \sum_{\underset{k = 1}{\overset{2 j + 1}{\sum}} n_k = n - 2 j - 1} C^{2 j + 1} \times\\
& \times (K\log (2 j + 1))^{m - n} \prod_{k = 1}^{2 j + 1} \| (1 + \cU)^{m - n} \nu_{n_k} \|_{\cC}.
\end{align*}
Since
$\| (1 + \cU)^{m - n} \nu_{n_k} \|_{\cC} \leq \| (1 + \cU)^{m - n_k} \nu_{n_k} \|_{L^2 \cap \cC}$,
we arrive at the inequality
\begin{equation}
\label{prince}
J_{m, n} \leq J_{m, n - 1} + \sum_{j = 3}^n C^j (K \log j)^{m - n} \sum_{\underset{k = 1}{\overset{j}{\sum}} n_k = n - j}  \bigg( \prod_{k = 1}^j J_{m, n_k} \bigg).
\end{equation}
Combining \eqref{port} with \eqref{prince}, we obtain
\begin{multline*}
P_m(t) \leq (1 + t + t^2) \| (1 + \cU)^m \nu \|_{L^2 \cap \cC} + t P_m(t)\\
+ \sum_{n = 3}^m \sum_{j = 3}^n C^j t^j (K \log j )^{m - n} \sum_{\underset{k = 1}{\overset{j}{\sum}} n_k = n - j}  \bigg( \prod_{k = 1}^j J_{m, n_k} t^{n_k} \bigg).
\end{multline*}
By Fubini's theorem, it holds
\begin{multline}
\label{preval}
P_m(t) \leq (1 + t + t^2) \| (1 + \cU)^m \nu \|_{L^2 \cap \cC} + t P_m(t)\\
+ \sum_{j = 3}^m C^j t^j (K\log j )^{m - j} \sum_{p = 0}^{m - j} \sum_{\underset{k = 1}{\overset{j}{\sum}} n_k = p}  \bigg( \prod_{k = 1}^j J_{m, n_k} t^{n_k} \bigg),
\end{multline}
Noticing that 
$$\sum_{p = 0}^{m - j} \sum_{\underset{k = 1}{\overset{j}{\sum}} n_k = p}  \bigg( \prod_{k = 1}^j J_{m, n_k} t^{n_k} \bigg) \leq P_m(t)^j,$$
we deduce \eqref{haiti} from \eqref{preval}. This completes the proof of Step \ref{P1}.
In the second step of the proof of Theorem \ref{Thm:prop_nu_n} we compute suitable bounds on $\cQ_m$ near the origin.

\begin{step}
\label{P2}
Let $m \geq 3$. There exists a positive constant $A(C, K)$, depending on $C$ and $K$, but not on $m$, such that
\begin{equation}
\label{estim:Q_m}
\cQ_m(u) \leq 2 u,\qquad \text{for any}\quad 0 \leq u \leq U_m := \frac{A(C, K)}{(K\log(m))^{m/2}}.
\end{equation}
\end{step}

By the definition \eqref{def:Q_m} of $\cQ_m$, we have
$$\cQ_m(u) \leq u + \sum_{j = 3}^m C^j (K\log(j))^{m - j} u^j \leq u + (K\log(m))^m \sum_{j = 3}^m \Big( \frac{C u}{K\log m} \Big)^j,$$
hence when $2 C u \leq K\log m$ and $2 C^3 (K \log(m))^{m - 3} u^2 \leq 1$, it holds $\cQ_m(u) \leq 2u$.
This ends the proof of Step \ref{P2}.

\begin{step}
\label{P3}
Conclusion of the proof of Theorem \ref{Thm:prop_nu_n}.
\end{step}

Since the coefficients $J_{m, n}$ are non negative, the function $t \mapsto t P_m(t)$ is either identically equal to $0$ (then \eqref{estim:nu_n} is straightforward), or increasing on $\R^+$. In the second case, it tends to $\ii$ as $t \to \ii$, and there exists a unique $T_m>0$ such that
\begin{equation}
\label{def:T_m}
T_m P_m(T_m) = U_m.
\end{equation}
Two situations may then occur. If $T_m \geq 1/4$, by \eqref{haiti} and \eqref{estim:Q_m},
$$P_m(t) \leq 2 \| (1 + \cU)^m \nu \|_{L^2 \cap \cC} + 2 t P_m(t)$$
for all $0 \leq t \leq 1/4$. Hence
$P_m(t) \leq 4 \| (1 + \cU)^m \nu \|_{L^2 \cap \cC}$ and
$$J_{m, n} \leq 4^n P_m \Big( \frac{1}{4} \Big) \leq 4^{n + 1} \| (1 + \cU)^m \nu \|_{L^2 \cap \cC}.$$
Otherwise $T_m \leq {1}/{4}$ and in this case we can deduce from \eqref{haiti}, \eqref{estim:Q_m} and \eqref{def:T_m} that ${U_m}/{T_m} \leq 2 \| (1 + \cU)^m \nu \|_{L^2 \cap \cC} + 2 U_m$.
This gives
$$T_m \geq \frac{U_m}{4 \| (1 + \cU)^m \nu \|_{L^2 \cap \cC}}.$$
Combining with \eqref{def:T_m} again, we are led to
$$J_{m, n} \leq \frac{U_m}{T_m^{n + 1}} \leq \frac{4^{n + 1}}{U_m^n} \| (1 + \cU)^m \nu \|_{L^2 \cap \cC}^{n + 1}.$$
Estimate \eqref{estim:nu_n} then follows from \eqref{estim:Q_m}.
\end{proof}

We now turn to the proof of Proposition \ref{Prop:nu_nL}.

\begin{proof}[Proof of Proposition \ref{Prop:nu_nL}]
The proof is almost identical. Denoting
$J_{m, n}^\Lambda := \| (1 + \cU)^{m - n} \nu_{n, \Lambda} \|_{L^2 \cap \cC}$
and introducing the polynomial function $P_m^\Lambda(t)$ given by
$$P_m^\Lambda(t) := \sum_{n = 0}^m J_{m, n}^\Lambda t^n,$$
we deduce from the definition \eqref{def:nu_nL}, and from \eqref{elise} and \eqref{estim:FnL} that
\begin{equation}
\label{domingue}
P_m^\Lambda(t) \leq \big( 1 + \kappa_1 t + \kappa_1^2 t^2) \| (1 + \cU)^m \nu \|_{L^2 \cap \cC} + (\kappa_1 - 1) t P_m^\Lambda(t) + \cQ_m(t P_m^\Lambda(t)),
\end{equation}
for all $t \geq 0$. Estimate \eqref{estim:nu_nL} then follows by applying to \eqref{domingue} the arguments of Steps \ref{P2} and \ref{P3} of the proof of Theorem \ref{Thm:prop_nu_n}.
\end{proof}

\section{Proof of Theorem \ref{Thm:asymptotics}}
\label{sec:proof_thm_asymptotics}

This last section is devoted to the proof of our main estimate \eqref{eq:asymptotics}. The proof relies on the identity
\begin{equation}
\label{eq:dvt_m}
\rhoph = \nu_\Lambda + \alphaph U_\Lambda \rhoph - \sum_{3 \leq 2 n + 1 \leq N} \alphaph^{2 n + 1} F_{2 n + 1, \Lambda}(\rhoph, \cdots, \rhoph) - \alphaph^{N + 1} G_{N + 1, \Lambda},
\end{equation}
where we denote
\begin{multline}
\label{def:G_mL}
G_{N + 1, \Lambda} := \rho \Bigg( \frac{1}{2\pi} \int_{-\infty}^{\ii} \frac{1}{D^0 - \alphaph \rhoph \ast |\cdot|^{- 1} + i \eta} \times\\
\times \prod_{j = 1}^{N + 1} \Big( \Pi_\Lambda \big( \rhoph \ast \frac{1}{|\cdot|} \big) \Pi_\Lambda \frac{1}{D^0 + i\eta} \Big) d\eta \Bigg).
\end{multline}
The formula \eqref{eq:dvt_m} follows from Cauchy's formula applied to \eqref{eq:SCF_recast2}. As mentioned in the introduction, the proof of \eqref{eq:asymptotics} naturally splits into two steps: we first establish that, under the assumptions of Theorem \ref{Thm:asymptotics}, the error term
\begin{equation}
\label{def:R_m}
R_N(\alphaph,\Lambda) := \rhoph(\alphaph,\Lambda) - \sum_{n = 0}^N \nu_{n, \Lambda} \alphaph^n.
\end{equation}
is controlled by a factor $\alphaph^{N + 1}$ (up to some multiplicative constant depending only on $N$, $\nu$ and $\epsilon$). In a second step we estimate the differences $\nu_{n, \Lambda} - \nu_n$ and deduce \eqref{eq:asymptotics}.
More precisely, the remainder $R_N$ satisfies the following

\begin{lemma}
\label{Lem:reste}
Let $N \in \N$ and $0 < \epsilon < 1$. Assume that $\epsilon \leq Z_3=1-\alphaph B_\Lambda \leq 1 - \epsilon$ and $\cN_N := \| (1 + \cU)^{N + 1} \nu \|_{L^2  \cap \cC} < \ii$. Then, there exist two constants $C(m, \epsilon, \cN_N)$ and $a(N, \epsilon, N_N)$, depending only on $N$, $\epsilon$ and $\cN_N$, such that
\begin{equation}
\label{estim:R_m}
\Big\| R_N(\alphaph,\Lambda) \Big\|_{L^2 \cap \cC} \leq C(N, \epsilon, \cN_N)\; \alphaph^{N + 1},
\end{equation}
for all $0 \leq \alphaph \leq a(N, \epsilon, \cN_N)$.
\end{lemma}

As for the differences $\nu_{n, \Lambda} - \nu_n$, we have the

\begin{lemma}
\label{Lem:diff_nu}
Let $\Lambda \geq 1$ and $N \in \N$. Assume that $\cN_N := \| (1 + \cU)^{N + 1} \nu \|_{L^2  \cap \cC} < \ii$. Then, there exists a constant $C(N, \cN_N)$, depending only on $N$ and $\cN_N$, such that
\begin{equation}
\label{estim:diff}
\Big\| \nu_{n, \Lambda} - \nu_n \Big\|_{L^2 \cap \cC} \leq \frac{C(N, \cN_N)}{(1 + B_\Lambda)^{N + 1 - n}},
\end{equation}
for all $0 \leq n \leq N$.
\end{lemma}

Combining Lemmas \ref{Lem:reste} and \ref{Lem:diff_nu}, we can complete the proof of Theorem \ref{Thm:asymptotics}.

\begin{proof}[Proof of Theorem \ref{Thm:asymptotics}]
Our assumption \eqref{condition_nu} (together with \eqref{UE}) means that $\cN_N := \| (1 + \cU)^{N + 1} \nu \|_{L^2  \cap \cC} < \ii$.
It follows from \eqref{def:R_m} that
$$\rhoph(\alphaph,\Lambda) - \sum_{n = 0}^N \nu_n (\alphaph)^n = R_N(\alphaph,\Lambda) + \sum_{n = 0}^N \big( \nu_{n, \Lambda} - \nu_n \big) (\alphaph)^n.$$
Hence by \eqref{estim:R_m} and \eqref{estim:diff},
\begin{equation}
\label{thatsit}
\Big\| \rhoph(\alphaph) - \sum_{n = 0}^N \nu_n \alphaph^n \Big\|_{L^2 \cap \cC} \leq C(N, \epsilon, \cN_N) \alphaph^{N + 1} + C(N, \cN_N) \sum_{n = 0}^N \frac{\alphaph^n}{(1 + B_\Lambda)^{N + 1 - n}},
\end{equation}
for any number $\alphaph$ sufficiently small. In our setting we have $B_\Lambda \geq \epsilon/\alphaph$ and the result follows. 
\end{proof}

It therefore remains to show Lemmas \ref{Lem:reste} and \ref{Lem:diff_nu}.

\begin{proof}[Proof of Lemma \ref{Lem:reste}]
Let us introduce the notation
\begin{equation}
\label{def:r_m}
r_N(\alphaph) := (\alphaph)^{-N- 1} R_N(\alphaph).
\end{equation}
We want to establish a bound on $r_N$ independently of $\alphaph$. By \eqref{def:R_m}, this requires to estimate $\rhoph$ and $\nu_{n, \Lambda}$ (which was already done in Proposition \ref{Prop:nu_nL}).

The first step of the proof will be to bound $\rhoph$ independently of $\alphaph$. Let us recall that a ground state for the reduced Bogoliubov-Dirac-Fock model satisfies
$$\| \alphaph \rhoph \|_\cC = \| \alpha (\rho_Q - \nu) \|_\cC \leq \alpha \| \nu \|_\cC$$
(see \cite[Eq. (33)]{HaiLewSer-05b}). Since $\alphaph=Z_3\alpha$, this provides
\begin{equation}
\label{estim:rhoph}
\| \rhoph \|_\cC \leq Z_3^{- 1} \| \nu \|_\cC \leq \epsilon^{- 1} \| \nu \|_{\cC}.
\end{equation}
Note that we however do not have any \emph{a priori} bound in $L^2(\R^3)$.
Inserting \eqref{def:R_m} and \eqref{def:r_m} in \eqref{eq:dvt_m} and using \eqref{def:nu_nL}, we get
\begin{multline}
\label{eq:r_m}
r_N = \alphaph \cU_\Lambda r_N + \cU_\Lambda \nu_{N, \Lambda} + G_{N + 1, \Lambda} + \sum_{k = N + 1}^{N(N + 2)} \alphaph^{k - N - 1} \times\\
\times \sum_{3 \leq 2 n + 1 \leq N}\ \sum_{p_1 + \cdots + p_{2 n + 1}  = k - 2 n - 1} F_{2 n + 1, \Lambda}(\omega_{p_1}, \cdots, \omega_{p_{2 n + 1}}),
\end{multline}
where $\omega_p = \nu_{p, \Lambda}$ for $0 \leq p \leq N$, and $\omega_{N + 1} = r_N$. It rests to estimate all the terms of the right-hand side of \eqref{eq:r_m}.

For the first term, we recall that $\alphaph |U_\Lambda| \leq \alphaph B_\Lambda = 1 - Z_3 \leq 1 - \epsilon$, therefore
\begin{equation}
\label{cayes}
\| \alphaph \cU_\Lambda r_N \|_{L^2 \cap \cC} \leq (1 - \epsilon) \| r_N \|_{L^2 \cap \cC}.
\end{equation}
The second term can be controlled by using \eqref{elise} and \eqref{estim:nu_nL}, which provide a positive constant $C(N, \cN_N)$, depending only on $N$ and $\cN_N$, such that
\begin{equation}
\label{saint-marc}
\| \cU_\Lambda \nu_{N, \Lambda} \|_{L^2 \cap \cC} \leq \kappa_1 \| (1 + \cU) \nu_{N, \Lambda} \|_{L^2 \cap \cC} \leq C(N, \cN_N).
\end{equation}
As for the function $G_{N + 1, \Lambda}$, we first recall that
\begin{equation}
\label{gonaives}
\Big( 1 - \frac{\pi^\frac{1}{6} 2^\frac{11}{6}}{\epsilon} \alphaph \norm{\nu}_\cC \Big) |D^0| \leq \Big| D^0 - \alphaph \rhoph \ast \frac{1}{|\cdot|} \Big| \leq \Big( 1 + \frac{\pi^\frac{1}{6} 2^\frac{11}{6}}{\epsilon} \alphaph \norm{\nu}_\cC \Big) |D^0|
\end{equation}
for all $\alphaph < \pi^{- 1/6} 2^{- 11/6} \epsilon \norm{\nu}_\cC^{-1}$ (see  \cite[p. 4495]{HaiLewSer-05b}). Hence, the operator $D^0 - \alphaph \rhoph \ast |\cdot|^{- 1}$ is invertible and, in particular, $G_{N +1, \Lambda}$ is well-defined.
Notice also that \eqref{gonaives} yields for any $\alphaph < \pi^{- 1/6} 2^{- 17/6} \epsilon \norm{\nu}_\cC^{-1}$
$$\frac{1}{2} |D^0| \leq \Big| D^0 - \alphaph \rhoph \ast \frac{1}{|\cdot|} \Big| \leq \frac{3}{2} |D^0|.$$

When $N \geq 5$, we argue exactly as in Steps \ref{F1} and \ref{F2} of the proof of Proposition \ref{Prop:estimF}, and deduce that there exists a constant $C(N)$, depending only on $N$, such that
\begin{equation}
\label{hinche}
\| G_{N + 1, \Lambda} \|_{L^2 \cap \cC} \leq C(N) \| \rhoph \|^{N + 1}_\cC.
\end{equation}
When $N \leq 4$, our argument is different. We expand $G_{N + 1, \Lambda}$ as before, writting
$$G_{N + 1, \Lambda} = - \alphaph \sum_{N + 1 \leq 2 j + 1 \leq 5} F_{2 j + 1, \Lambda}(\rhoph, \cdots, \rhoph) + G_{6, \Lambda}.$$
In view of \eqref{estim:FnL} and \eqref{hinche} (for $N = 5$), this leads to
$$\| G_{N + 1, \Lambda} \|_{L^2 \cap \cC} \leq C \alphaph \left( \sum_{N + 1 \leq 2 j + 1 \leq 5} \| \rhoph \|^{2 j + 1}_\cC + \| \rhoph \|^{6}_\cC \right).$$
In both cases, we obtain
$$\| G_{N + 1, \Lambda} \|_{L^2 \cap \cC} \leq C(N) \max \left\{ \| \rhoph \|^{N + 1}_\cC, \| \rhoph \|^6_\cC \right\},$$
for any $\alphaph \leq 1$, so that, by \eqref{estim:rhoph},
\begin{equation}
\label{jeremie}
\| G_{N + 1, \Lambda} \|_{L^2 \cap \cC} \leq C(N) \max \left\{ \frac{\norm{\nu}^{N + 1}_\cC}{\epsilon^{N + 1}}, \frac{ \norm{\nu}^6_\cC}{\epsilon^6} \right\} \leq C(N, \cN_N, \epsilon).
\end{equation}
Finally, we consider the terms $\alphaph^{k - m - 1} F_{2 n + 1, \Lambda}(\omega_{p_1}, \cdots, \omega_{p_{2 n + 1}})$ of the sum in the right-hand side of \eqref{eq:r_m}. By \eqref{estim:FnL}, we have
$$\big\| \alphaph^{k - N - 1} F_{2 n + 1, \Lambda}(\omega_{p_1}, \cdots, \omega_{p_{2 n + 1}}) \big\|_{L^2 \cap \cC} \leq C^{2 n + 1} |\alphaph|^{k - N - 1} \prod_{j = 1}^{2 n + 1} \| \omega_{p_j} \|_{\cC}.$$
When $p_j \leq N$, we deduce from \eqref{estim:nu_nL} that there exists a constant $C(N, \cN_N)$ such that $\| \omega_{p_j} \|_{\cC} \leq C(N, \cN_N)$.
On the other hand, when $p_j = N + 1$ for some $j$, we can bound one of the norms $\| \omega_{p_j} \|_{\cC}$ by $\| r_N \|_{\cC}$, and the other ones by using \eqref{estim:nu_nL}, \eqref{def:R_m} and \eqref{estim:rhoph} to get
$$\| \omega_{p_j} \|_{\cC} = |\alphaph|^{- p_j} \| R_N \| \leq C(N, \cN_N, \epsilon) |\alphaph|^{- p_j},$$
for $\alphaph \leq 1$. This leads to
$$\big\| \alphaph^{k - N - 1} F_{2 n + 1, \Lambda}(\omega_{p_1}, \cdots, \omega_{p_{2 n + 1}}) \big\|_{L^2 \cap \cC} \leq C(N, \cN_N, \epsilon) |\alphaph|^{2 n + 1} \max \{ \| r_N \|_{\cC}, 1 \},$$
for $\alphaph \leq 1$. Combining with \eqref{eq:r_m}, \eqref{cayes}, \eqref{saint-marc} and \eqref{jeremie}, we conclude that
$$\| r_N \|_{L^2  \cap \cC} \leq C(N, \cN_N, \epsilon) + \bigg( 1 - \epsilon + C(N, \cN_N, \epsilon) |\alphaph|^3 \bigg) \max \{ \| r_N \|_{\cC}, 1 \},$$
for $\alphaph$ sufficiently small. Therefore, the norm $\norm{r_N}_{L^2 \cap \cC}$ is bounded independently of $\alphaph$ for $\alphaph$ small enough, which ends the proof of Lemma \ref{Lem:reste}.
\end{proof}

We finally prove Lemma \ref{Lem:diff_nu}.

\begin{proof}[Proof of Lemma \ref{Lem:diff_nu}]
Given any $n \in \{0, 1, 2 \}$, it follows from recursion relations \eqref{def:nu_nL} and \eqref{def:nu_n} that
$$\nu_{n, \Lambda} - \nu_n = \cU_\Lambda^n \nu_\Lambda - \cU^n \nu = \cU_\Lambda^n \big( \nu_\Lambda - \nu \big) + \big( \cU_\Lambda^n - \cU^n) \nu.$$
Therefore, given any $N \geq n$ and $0 \leq p \leq N + 1 - n$, we deduce from \eqref{elise} that
\begin{equation}
\label{bonnaire}
\begin{split}
\| (1 + \cU)^p (\nu_{n, \Lambda} - \nu_n) \|_{L^2 \cap \cC} \leq \kappa_1^n \| (1 & + \cU)^{n + p} (\nu_\Lambda - \nu) \|_{L^2 \cap \cC}\\
+ & n \kappa_1^{n - 1} \| (1 + \cU)^{n + p - 1} \big( \cU_\Lambda - \cU) \nu \|_{L^2 \cap \cC}.
\end{split}
\end{equation}
Next, we recall that $\widehat{\nu_\Lambda} = \widehat{\nu} \1_{B(0,2\Lambda)}$, so that, since $U(2 \Lambda) = B_\Lambda$,
\begin{equation}
\label{cudmore}
\norm{(1 + \cU)^{n + p} (\nu_{\Lambda} - \nu)}_{L^2\cap\cC} \leq \frac{\norm{(1 + \cU)^{N + 1} \nu}_{L^2 \cap \cC}}{(1 + U(2 \Lambda))^{N + 1 - n - p}} = \frac{\norm{(1 + \cU)^{N + 1} \nu}_{L^2 \cap \cC}}{(1 + B_\Lambda)^{N + 1 - n - p}}.
\end{equation}
For the second term in the right-hand side of \eqref{bonnaire}, we use \eqref{estim:ULU} and write
\begin{align*}
\| (1 + & \cU)^{n + p - 1} \big( \cU_\Lambda - \cU) \nu \|_{L^2 \cap \cC} \leq \norm{\frac{U_\Lambda - U}{(1 + U)^{N + 2 - n -p}}}_{L^\infty} \norm{(1 + \cU)^{N + 1} \nu}_{L^2 \cap \cC}\\
& \leq \kappa_0^{N + 4 - n - p} \max \Big\{ \frac{1}{(1 + B_\Lambda)^{N + 1 - n - p}}, \frac{1}{E(2 \Lambda)} \Big\} \norm{(1 + \cU)^{N + 1} \nu}_{L^2 \cap \cC}.
\end{align*}
Since $(1 + B_\Lambda)^{N + 1 - n - p} \leq (1 + B_\Lambda)^{N + 1} \leq C(N) E(2 \Lambda)$, we obtain
$$\Big\| (1 + \cU)^p (\nu_{n, \Lambda} - \nu_n) \Big\|_{L^2 \cap \cC} \leq \frac{C(N)}{(1 + B_\Lambda)^{N + 1 - n - p}} \norm{(1 + \cU)^{N + 1} \nu}_{L^2 \cap \cC}.$$
Combining with \eqref{bonnaire} and \eqref{cudmore}, we are led to
\begin{equation}
\label{vermeulen}
\Big\| (1 + \cU)^p (\nu_{n, \Lambda} - \nu_n) \Big\|_{L^2 \cap \cC} \leq \frac{C(N, \cN_N)}{(1 + B_\Lambda)^{N + 1 - n - p}},
\end{equation}
for $N \geq n$ and $0 \leq p \leq N + 1 - n$.

We next turn to the case of $n \geq 3$. Given any $N \geq n$, we assume that \eqref{vermeulen} holds for all $n \leq k - 1$ and $0 \leq p \leq N + 1 - n$, and prove it by induction for $n = k$ and $0 \leq p \leq N + 1 - k$. Using \eqref{def:nu_nL} and \eqref{def:nu_n}, we first infer that
\begin{multline}
\label{parra}
\big\| (1 + \cU)^p (\nu_{k, \Lambda} - \nu_k) \big\|_{L^2 \cap \cC} \leq \big\| (1 + \cU)^p \cU_\Lambda (\nu_{k - 1, \Lambda} - \nu_{k - 1}) \big\|_{L^2 \cap \cC}\\
+ \big\| (1 + \cU)^p (\cU_\Lambda  - \cU) \nu_{k - 1} \big\|_{L^2\cap\cC} + \sum_{3 \leq 2 j + 1 \leq k}\ \sum_{\underset{\ell = 1}{\overset{2 j + 1}{\sum}} k_\ell = k - 2 j - 1} \Big\| (1 + \cU)^p \times\\
\times \Big( F_{2 j + 1, \Lambda}  \big( \nu_{k_1, \Lambda}, \cdots, \nu_{k_{2 j + 1}, \Lambda} \big) - F_{2 j + 1} \big(\nu_{k_1}, \cdots, \nu_{k_{2 j + 1}} \big) \Big) \Big\|_{L^2 \cap \cC}.
\end{multline}
We next estimate the first term in the right-hand side of \eqref{parra} using \eqref{elise} and our assumption. This provides
\begin{multline}
\label{james}
\| (1 + \cU)^p \cU_\Lambda (\nu_{k - 1, \Lambda} - \nu_{k - 1}) \|_{L^2 \cap \cC} \leq \kappa_1 \| (1 + \cU)^{p + 1} (\nu_{k - 1, \Lambda} - \nu_{k - 1}) \big\|_{L^2 \cap \cC}\\
\leq \frac{C(N)}{(1 + B_\Lambda)^{N + 1 - k - p}} \norm{(1 + \cU)^{N + 1} \nu}_{L^2 \cap \cC}.
\end{multline}
For the second term, we argue as in the proof of \eqref{vermeulen}, using \eqref{estim:ULU} and \eqref{estim:nu_n}:
\begin{align}
\label{nalaga}
\| (1 + \cU)^p (\cU_\Lambda - \cU) \nu_{k - 1} \|_{L^2\cap\cC} &\leq \norm{\frac{U_\Lambda - U}{(1 + U)^{N + 2 - k - p}}}_{L^\infty} \norm{(1 + \cU)^{N + 2 - k} \nu_{k - 1}}_{L^2 \cap \cC}\\
& \leq \frac{C(N, \cN_N)}{(1 + B_\Lambda)^{N + 1 - k - p}}.\nonumber
\end{align}
Finally, we turn to the terms in the sums of the right-hand side of \eqref{parra}. On the one hand, we deduce from \eqref{def:FnLe} that
$$F_{2 j + 1, \Lambda} - F_{2 j + 1} = - F_{2 j + 1, \Lambda}^{(- 1, 1, \cdots, 1)} - F_{2 j + 1, \Lambda}^{(0, - 1, 1, \cdots, 1)} - \cdots - F_{2 j + 1, \Lambda}^{(0, \cdots, 0, - 1)}.$$
Hence, since $p \leq N + 1 - k \leq N + 1 - k_\ell$, we can apply \eqref{estim:F} and \eqref{estim:nu_nL} to obtain
\begin{multline}
\label{jacquet}
\Big\| (1 + \cU)^p \Big( F_{2 j + 1, \Lambda} \big( \nu_{k_1, \Lambda}, \cdots, \nu_{k_{2 j + 1}, \Lambda} \big) - F_{2 j + 1} \big( \nu_{k_1, \Lambda}, \cdots, \nu_{k_{2 j + 1}, \Lambda} \big) \Big) \Big\|_{L^2 \cap \cC}\\
\leq \frac{C(N)}{\Lambda^{1/24}} \prod_{\ell = 1}^{2 j + 1} \| (1 + \cU)^{N + 1 - k_\ell} \nu_{k_\ell, \Lambda} \|_{\cC} \leq \frac{C(N, \cN_N)}{\Lambda^{1/24}}.
\end{multline}
On the other hand, the multilinearity of the function $F_{2 j + 1}$ provides
\begin{align*}
& F_{2 j + 1} \big( \nu_{k_1, \Lambda}, \cdots, \nu_{k_{2 j + 1}, \Lambda} \big) -F_{2 j + 1} \big(\nu_{k_1}, \cdots, \nu_{k_{2 j + 1}} \big)\\
= F_{2 j + 1} \big( \nu_{k_1, \Lambda} & - \nu_{k_1}, \nu_{k_2, \Lambda}, \cdots, \nu_{k_{2 j + 1}, \Lambda} \big) + F_{2 j + 1} \big( \nu_{k_1}, \nu_{k_2, \Lambda} - \nu_{k_2}, \cdots, \nu_{k_{2 j + 1}, \Lambda} \big)\\
& + \cdots + F_{2 j + 1} \big( \nu_{k_1}, \nu_{k_2}, \cdots, \nu_{k_{2 j}}, \nu_{k_{2 j + 1}, \Lambda} - \nu_{k_{2 j + 1}} \big).
\end{align*}
Therefore, we infer similarly from \eqref{estim:F} and \eqref{estim:nu_nL} that
\begin{align}
\label{ledesma}
& \Big\| (1 + \cU)^p \Big( F_{2 j + 1} \big( \nu_{k_1, \Lambda}, \cdots, \nu_{k_{2 j + 1}, \Lambda} \big) - F_{2 j + 1} \big( \nu_{k_1, \Lambda}, \cdots, \nu_{k_{2 j + 1}, \Lambda} \big) \Big) \Big\|_{L^2 \cap \cC}\\
&\ \leq C(N)  \sum_{q = 1}^{2 j + 1} \big\| (1 + \cU)^p (\nu_{k_q, \Lambda} - \nu_{k_q}) \big\|_{\cC}  \prod_{\ell < q} \big\| (1 + \cU)^p \nu_{k_\ell, \Lambda} \big\|_{\cC}  \prod_{\ell > q} \big\| (1 + \cU)^p \nu_{k_\ell} \big\|_{\cC} \nonumber\\
&\ \leq \sum_{q = 1}^{2 j + 1} \frac{C(N, \cN_N)}{(1 + B_\Lambda)^{N + 1 - k_q - p}} \leq \frac{C(N, \cN_N)}{(1 + B_\Lambda)^{N + 1 - k - p}}.\nonumber
\end{align}
As a conclusion, we derive from \eqref{parra}, \eqref{james}, \eqref{nalaga}, \eqref{jacquet} and \eqref{ledesma} that
$$\Big\| (1 + \cU)^p (\nu_{k, \Lambda} - \nu_k) \Big\|_{L^2 \cap \cC} \leq C(N, \cN_N) \left( \frac{1}{(1 + B_\Lambda)^{N + 1 - k - p}} + \frac{1}{\Lambda^{1/24}}\right),$$
Since $(1 + B_\Lambda)^{N + 1 - k - p} \leq (1 + B_\Lambda)^{N + 1} \leq C(N) \Lambda^{1/24}$, this completes the proof of \eqref{vermeulen} for $n = k$.

Notice the constant $C(N, \cN_N)$ deteriorates when $n$ increases. However, this is not a problem since $n$ is limited to the set $\{ 0, \cdots, N \}$. Estimate \eqref{estim:diff} then follows from \eqref{vermeulen}, considering the case $p = 0$. This concludes the proof of Lemma \ref{Lem:diff_nu}.
\end{proof}

\appendix
\section{Auxiliary results on the Uehling multiplier $U$}\label{sec:appendix}

\subsection{Elementary properties of $U$}\label{sec:app:U}
We gather in this section some important properties of $U$, which will be useful for the proof of Lemma \ref{Lem:U-prod} in the next section.
\begin{lemma}
\label{Lem:U}
The function $U$ defined in \eqref{def:U} is a non-negative, non-decreasing, smooth function on $\R_+$ such that
\begin{equation}
\label{lim:U}
U(r) \underset{r \to 0}{\sim} \frac{r^2}{15 \pi} \quad \text{and} \quad U(r) \underset{r \to \ii}{\sim} \frac{2}{3 \pi} \log r.
\end{equation}
Its derivative $U'$ is positive on $(0, \ii)$, and it holds
\begin{equation}
\label{lim:dU_ddU}
U'(r) \underset{r \to \ii}{\sim} \frac{2}{3 \pi r}\quad \text{and} \quad   U''(r) \underset{r \to \ii}{\sim} - \frac{2}{3 \pi r^2}.
\end{equation}
Moreover, we have
\begin{equation}
\label{UE}
\forall r\in\R^+,\qquad \frac{2}{15 \pi} (1 + \log E(r)) \leq 1 + U(r) \leq 1 + \frac{2}{3 \pi} \log E(r).
\end{equation}
\end{lemma}

\begin{proof}[Proof of Lemma \ref{Lem:U}]
For the convenience of the reader, let us recall the integral and the explicit formulas \eqref{def:U} of $U$:
\begin{align}
\label{def:U_app}
U(r)&=\frac{r^2}{4\pi}\int_0^1\frac{z^2-z^4/3}{1+\frac{r^2(1-z^2)}{4}}\,dz\\
&\nonumber=\frac{12-5r^2}{9\pi r^2}+\frac{\sqrt{4+r^2}}{3\pi r^3}(r^2-2)\log\left(\frac{\sqrt{4+r^2}+r}{\sqrt{4+r^2}-r}\right).
\end{align}
Most of the statements of Lemma \ref{Lem:U} are direct consequences of \eqref{def:U_app}. 
As for \eqref{UE}, we estimate, using \eqref{def:U_app},
$$U(r) \leq \frac{r^2}{12 \pi} \int_0^1 \frac{2 z}{1 + \frac{r^2}{4} - \frac{r^2}{4} z^2} dz = \frac{1}{3 \pi} \log \Big( 1 + \frac{r^2}{4} \Big) \leq \frac{2}{3 \pi} \log E(r).$$
For the lower bound, we notice similarly that
$$U(r) \geq \frac{r^2}{4 \pi \big( 1 + \frac{r^2}{4} \big)} \int_0^1 \Big( z^2 - \frac{z^4}{3} \Big) dz = \frac{r^2}{15 \pi \big( 1 + \frac{r^2}{4} \big)},$$
for $r \in \R_+$, so that
\begin{equation}
\label{misan}
\forall 0 \leq r \leq 1,\qquad \frac{2}{15 \pi} \big( 1 + \log E(r) \big) \leq \frac{2 + r^2}{15 \pi} \leq 1 + U(r).
\end{equation}
On the other hand, we can also write
$$U(r) \geq \frac{r^2}{6 \pi} \int_0^1 \frac{z^2}{1 + \frac{r^2}{2}(1 - z)} dz = \frac{4}{3 \pi r^4} \bigg( \Big( 1 + \frac{r^2}{2} \Big)^2 \log \Big( 1  + \frac{r^2}{2} \Big) - \frac{r^2}{2} - \frac{r^4}{8} \bigg),$$
thus when $r > 1$
$$1 + U(r) \geq \frac{1}{3 \pi} \log \Big( 1  + \frac{r^2}{2} \Big) + 1 - \frac{7}{6 \pi} \geq \frac{1}{3 \pi} \big( 1 + \log E(r) \big).$$
The lower bound in \eqref{UE} then follows from \eqref{misan}.
\end{proof}

A useful consequence of Lemma \ref{Lem:U} is the following

\begin{lemma}
\label{Cor:Phi}
Let $\Phi$ be the function defined on $\R_+$ by
$$\Phi(r) = \frac{U'(r)}{1 + U(r)},$$
There exist three positive numbers $T_-$, $T_+$ and $\Phi_0$ such that the function $\Phi$ is an increasing diffeomorphism from $(0, T_-)$ onto $(0, \Phi_0)$, respectively a decreasing diffeomorphism from $(T_+, \ii)$ onto $(0, \Phi_0)$, and
$\Phi^{- 1} \big( (0, \Phi_0) \big) = (0, T_-) \cup (T_+, \ii)$.
Moreover, we have
\begin{equation}
\label{lim:Phi}
\Phi(r) \underset{r \to 0}{\sim} \frac{2 r}{15 \pi}, \ {\rm and} \ \Phi(r) \underset{r \to \ii}{\sim} \frac{1}{r \log r}.
\end{equation}
\end{lemma}

\begin{proof}[Proof of Lemma \ref{Cor:Phi}]
From Lemma \ref{Lem:U}, we see that the function $\Phi$ is well-defined, smooth on $\R_+$, and satisfies \eqref{lim:Phi}. Then we compute for $r \geq 0$:
$$\Phi'(r) = \frac{U''(r) (1 + U(r)) - U'(r)^2}{(1 + U(r))^2}.$$
By \eqref{lim:U} and \eqref{lim:dU_ddU}, we thus have $\Phi'(0) = \frac{2}{15 \pi}$ and $\Phi'(r) \sim_{r \to \ii} - {1}/(r^2 \log r)$.
Since $\Phi(0) = 0$ and $\Phi(r) \to 0$ as $r \to \ii$ by \eqref{lim:U} and \eqref{lim:dU_ddU}, there exist $a,b,\delta>0$ such that $\Phi$ is an increasing diffeomorphism from $(0, a)$ onto $(0, \delta)$, respectively a decreasing diffeomorphism from $(b, \ii)$ onto $(0, \delta)$. The function $\Phi$ is positive on $[a, b]$, so that
$m = \min \{ \Phi(t), a \leq t \leq b \} > 0$.
Lemma \ref{Cor:Phi} follows by introducing $\Phi_0 = \min \{ m/2, \delta \}$, and $T_- < T_+$, the two positive numbers such that $\Phi(T_-) = \Phi(T_+) = \Phi_0$.
\end{proof}

\subsection{A useful bound involving $U$}\label{sec:lem_U-prod}
We use here results from the previous section to derive a bound useful for the proof of Proposition \ref{Prop:estimF}.
\begin{lemma}
\label{Lem:U-prod}
There exists a universal constant $K >0$ such that
\begin{equation}
\label{U:prod}
 1 + U \left( \Big| \sum_{j = 1}^n v_j \Big| \right)  \leq K \log n \prod_{j = 1}^n \Big( 1+ U(|v_j|) \Big)
\end{equation}
for all $n \geq 1$, and all $(v_1, \cdots, v_n) \in (\R^3)^n$.
\end{lemma}

If we allow $K$ to depend on $n$, the optimal constant in the above inequality satisfies $K_n\to 1/3\pi$ when $n\to\ii$, as can be seen from the proof. The factor $\log n$ in \eqref{U:prod} is therefore optimal with regard to the large-$n$ dependence.

\begin{proof}[Proof of Lemma \ref{Lem:U-prod}]
By Lemma \ref{Lem:U}, it holds
\begin{equation}
\frac{1 + U \Big( \Big| \underset{j = 1}{\overset{n}{\sum}} v_j \Big| \Big)}{\underset{j = 1}{\overset{n}{\prod}} \big(1 + U(|v_j|) \big)} \leq \frac{1 + U \Big( \underset{j = 1}{\overset{n}{\sum}} |v_j| \Big)}{\underset{j = 1}{\overset{n}{\prod}} \big(1 + U(|v_j|) \big)} \leq 
\max_{t_1,...,t_n \in \R_+} \frac{1 + U \Big( \underset{j = 1}{\overset{n}{\sum}} t_j \Big)}{\underset{j = 1}{\overset{n}{\prod}} \big(1 + U(t_j) \big)}:=J_n.
\label{madmax}
\end{equation}
It is clear that taking $v_1=\cdots =v_n=v$ shows that the maximum of the left-hand side of \eqref{madmax} is actually $J_n$.
Next, we take $t_1=...=t_n=\tau_n$ in \eqref{madmax} with $\tau_n=\sqrt{15\pi/(n\log n)}$. Using \eqref{lim:dU_ddU}, we see that $J_n\gtrsim (\log n)/3\pi$ for $n\gg1$. We will show that actually it holds $J_n\sim (\log n)/(3\pi)$ when $n\to\ii$. In the rest of the proof, we assume $n\geq n_0$ is such that $J_n>1$.

Let us consider a maximizing sequence $\{(t_1^{(p)}, \cdots, t_n^{(p)})\}_{p \in \N}$ for the variational problem defining $J_n$. If the sequence is unbounded, then by Lemma \ref{Lem:U},
$$J_n \leq \lim_{p \to \ii} \frac{1 + U\left(n \max_j\{t^{(p)}_j\}\right)}{1 + U\left(\max_j\{t^{(p)}_j\}\right)} = 1$$
which contradicts $J_n>1$ for $n\geq n_0$. Therefore $(t^{(p)}_1,...,t^{(p)}_n)$ is bounded in $(\R_+)^n$.
In this case the variational problem on the right-hand side of \eqref{madmax} has a maximizer, which satisfies the equation
\begin{equation}
\label{mel}
\forall 1 \leq k \leq n,\qquad \Phi(t_k) = \frac{U'(t_k)}{1 + U(t_k)} = \frac{U' \Big( \underset{j = 1}{\overset{n}{\sum}} t_j \Big)}{1 + U \Big( \underset{j = 1}{\overset{n}{\sum}} t_j \Big)} = \Phi \left( \underset{j = 1}{\overset{n}{\sum}} t_j \right) := \Phi_1.
\end{equation}
Assume now that $\Phi_1 \geq \Phi_0.$
By Lemma \ref{Cor:Phi}, we have $ t_k \geq T_-$ and $\sum_{j=1}^n t_j\leq T_+$, for all $ 1 \leq k \leq n$, hence $n \leq T_+/T_-$. In particular, for $n > T_+/T_-$, it must hold $0 \leq \Phi_1 < \Phi_0$.
Note if $\Phi_1 = 0$, we infer from Lemma \ref{Lem:U} that $t_1 = \cdots = t_n = 0$, so that $J_n = 1$, a contradiction.
Therefore, by Lemma \ref{Cor:Phi}, there exist exactly two numbers $0 < \tau_n < T_-$ and $T_n > T_+$ such that
$\Phi(\tau_n) = \Phi(T_n) = \Phi_1.$
By \eqref{mel}, the unique possible maximizer is $(\tau_n, \cdots, \tau_n)$, where $\tau_n = T_n/n\in(0, T_-)$ is such that
\begin{equation}
\label{gibson}
\Phi(\tau_n) = \Phi(n \tau_n).
\end{equation}
The corresponding value of $J_n$ is 
\begin{equation}
\label{maxlamenace}
J_n = \frac{1 + U(n \tau_n)}{(1 + U(\tau_n))^n}.
\end{equation}
By \eqref{gibson}, we must have $\tau_n\to0$ as $n\to\ii$. Combining \eqref{gibson} with \eqref{lim:Phi}, it follows that
$\Phi(n \tau_n) \sim{2 \tau_n}/(15 \pi) \to0$.
By Lemma \ref{Cor:Phi} and since $n \tau_n = T_n \geq T_+$, it holds $n \tau_n\to\ii$.
Using \eqref{lim:Phi} again, we deduce that $n \tau_n \log(n \tau_n) \sim_{n \to \ii} 15 \pi/(2 \tau_n)$,
hence finally $\tau_n \sim \sqrt{{15 \pi}/(n \log n)}$.
Inserting in \eqref{lim:U} and \eqref{maxlamenace}, we finally arrive at $J_n \sim (\log n)/(3\pi)$.
This ends the proof of Lemma \ref{Lem:U-prod}.
\end{proof}

\section{Proof of Proposition \ref{Prop:ULU}} \label{sec:proof_Prop:ULU}
We start by showing the following lemma which provides estimates on $U - U_\Lambda$ on $[0, 2 \Lambda]$.

\begin{lemma}
\label{Lem:ULU}
Let $\Lambda \geq 1$. For $\kappa_1 = 258/\pi$, one has
\begin{equation}
\label{plaute}
\forall 0 \leq r \leq 2 \Lambda,\qquad |U_\Lambda(r) - U(r)| \leq \kappa_1 \frac{r}{2 E(\Lambda)}.
\end{equation}
\end{lemma}

\begin{proof}[Proof of Lemma \ref{Lem:ULU}]
Recall that (see \eqref{def:U} and \eqref{def:UL})
\begin{multline}
\label{diffU}
U_\Lambda(r) - U(r)  = \frac{r^2}{4 \pi} \int_\frac{\Lambda}{E(\Lambda)}^1 \frac{z^2 - \frac{z^4}{3}}{1 + \frac{r^2}{4} (1 - z^2)} dz - \frac{r}{2 \pi} \int_0^{Z_\Lambda(r)} \frac{z - \frac{z^3}{3}}{E(\Lambda) - \frac{r z}{2}}dz\\
 + \frac{1}{\pi} \int_{Z_\Lambda(r)}^{\frac{\Lambda}{E(\Lambda)}} \frac{z^2 - \frac{z^4}{3}}{(1 - z^2)(1 + \frac{r^2}{4}(1 - z^2))} dz,
\end{multline}
for $0 \leq r \leq 2 \Lambda$ and where
\begin{equation}
\label{def:ZL}
Z_\Lambda(r) = \frac{E(\Lambda) - E(\Lambda - r)}{r} = \frac{2 \Lambda - r}{E(\Lambda) + E(\Lambda -r)} \leq \frac{\Lambda}{E(\Lambda)}.
\end{equation}
We will estimate all the terms of the right-hand side of \eqref{diffU}.
The first term is treated as follows, for all $0 \leq r \leq 2 \Lambda$:
\begin{equation}
\bigg| \frac{r^2}{4 \pi}\int_\frac{\Lambda}{E(\Lambda)}^1 \frac{z^2 - \frac{z^4}{3}}{1 + \frac{r^2}{6} (1 - z^2)} dz \bigg| \leq \frac{r^2}{6 \pi} \Big( 1 - \frac{\Lambda}{E(\Lambda)} \Big) \leq \frac{r^2}{6 \pi E(\Lambda)^2}.
\label{aulu}
\end{equation}
Using \eqref{def:ZL} and $|x| \leq E(x)$, we bound the second term by
\begin{equation}
\label{laria}
\frac{r}{2 \pi} \bigg| \int_0^{Z_\Lambda(r)} \frac{z - \frac{z^3}{3}}{E(\Lambda) - \frac{r z}{2}}dz \bigg| \leq \frac{r Z_\Lambda(r)^2}{4 \pi \big( E(\Lambda) - \frac{r Z_\Lambda(r)}{2} \big)} 
\leq \frac{r}{2 \pi E(\Lambda)},
\end{equation}
for $0 \leq r \leq 2 \Lambda$.
In order to estimate the last term of the right-hand side of \eqref{diffU}, we distinguish the regions $0\leq r\leq\Lambda/2$ and $\Lambda/2\leq r\leq 2\Lambda$. We calculate
$$\bigg| \int_{Z_\Lambda(r)}^{\frac{\Lambda}{E(\Lambda)}} \frac{z^2 - \frac{z^4}{3}}{(1 - z^2)(1 + \frac{r^2}{4}(1 - z^2))} dz \bigg| \leq \frac{2}{3} \int_{Z_\Lambda(r)}^{\frac{\Lambda}{E(\Lambda)}} \frac{dz}{1 - z} = \frac{2}{3} \log \bigg( \frac{1 - Z_\Lambda(r)}{1 - \frac{\Lambda}{E(\Lambda)}} \bigg).$$
On the other hand, by \eqref{def:ZL},
$$\frac{1 - Z_\Lambda(r)}{1 - \frac{\Lambda}{E(\Lambda)}} = 1 + \frac{r(2 \Lambda -r) (\Lambda + E(\Lambda))}{(E(\Lambda) + E(\Lambda - r))((\Lambda -r) E(\Lambda) + \Lambda E(\Lambda -r))} \leq 1 + \frac{6 r}{E(\Lambda)},$$
as soon as  $0 \leq r \leq \Lambda/2$. Hence using $\log(1 + x) \leq x$ we infer the bound
\begin{equation}
\label{har}
\forall 0 \leq r \leq \Lambda/2,\qquad 
\frac{1}{\pi} \bigg| \int_{Z_\Lambda(r)}^{\frac{\Lambda}{E(\Lambda)}} \frac{z^2 - \frac{z^4}{3}}{(1 - z^2)(1 + \frac{r^2}{4}(1 - z^2))} dz \bigg| \leq \frac{4 r}{\pi E(\Lambda)}.
\end{equation}
For $\Lambda/2 \leq r \leq 2 \Lambda$, we write similarly as before
$$\bigg| \int_{Z_\Lambda(r)}^{\frac{\Lambda}{E(\Lambda)}} \frac{z^2 - \frac{z^4}{3}}{(1 - z^2)(1 + \frac{r^2}{4}(1 - z^2))} dz \bigg| \leq \frac{8}{3 r^2} \int_{Z_\Lambda(r)}^{\frac{\Lambda}{E(\Lambda)}} \frac{dz}{(1 - z)^2} \leq \frac{8}{3 r^2} E(\Lambda) (\Lambda + E(\Lambda))$$
and deduce the estimate
\begin{equation}
\label{pagon}
\forall \Lambda/2 \leq r \leq 2 \Lambda,\qquad
\frac{1}{\pi} \bigg| \int_{Z_\Lambda(r)}^{\frac{\Lambda}{E(\Lambda)}} \frac{z^2 - \frac{z^4}{3}}{(1 - z^2)(1 + \frac{r^2}{4}(1 - z^2))} dz \bigg| \leq \frac{128 r}{\pi E(\Lambda)}.
\end{equation}
Estimate \eqref{plaute} follows from \eqref{aulu}, \eqref{laria}, \eqref{har} and \eqref{pagon}, together  with \eqref{diffU}. This ends the proof of Lemma \ref{Lem:ULU}.
\end{proof}

We now use Lemma \ref{Lem:ULU} to finish the proof of Proposition \ref{Prop:ULU}.
The pointwise convergence of  $U_\Lambda$ when $\Lambda \to \ii$ is a direct consequence of \eqref{plaute}.
For \eqref{estim:ULU}, we first use \eqref{UE} and \eqref{plaute} to obtain
$$\forall 0 \leq r \leq 2 \Lambda,\qquad \bigg| \frac{U_\Lambda(r) - U(r)}{(1 + U(r))^{m + 1}} \bigg| \leq \kappa_1 \Big( \frac{15 \pi}{2} \Big)^{m + 1} \frac{E(r)}{2 E(\Lambda) (1 + \log E(r))^{m + 1}}.$$
Optimizing $x \mapsto \frac{E(x)}{(1 + \log E(x))^{m + 1}}$ on $[0, 2 \Lambda]$ yields
$$\frac{E(r)}{(1 + \log E(r))^{m + 1}} \leq \max \Big\{ 1, \frac{E(2 \Lambda)}{(1 + \log E(2 \Lambda))^{m + 1}} \Big\}.$$
Since $E(2 x) \leq 2 E(x)$ for any $x \geq 0$, we are led to
\begin{equation}
\label{cleante}
\bigg| \frac{U_\Lambda(r) - U(r)}{(1 + U(r))^{m + 1}} \bigg| \leq \kappa_1 \Big( \frac{15 \pi}{2} \Big)^{m + 1} \max \Big\{ \frac{1}{(1 + \log E(2 \Lambda))^{m + 1}}, \frac{1}{E(2 \Lambda)} \Big\}.
\end{equation}
On the other hand, $U$ is non-decreasing on $\R_+$, hence, using \eqref{UE} we infer
$$\forall r \geq 2 \Lambda,\qquad \bigg| \frac{U_\Lambda(r) - U(r)}{(1 + U(r))^{m + 1}} \bigg| = \bigg| \frac{U(r)}{(1 + U(r))^{m + 1}} \bigg| \leq \Big( \frac{15 \pi}{2} \Big)^m \frac{1}{(1 + \log E(2 \Lambda))^m}.$$
Using \eqref{cleante}, we finally obtain
\begin{equation}
\label{valere}
\bigg\| \frac{U_\Lambda - U}{(1 + U)^{m + 1}} \bigg\|_{L^\infty} \leq \kappa_1 \Big( \frac{15 \pi}{2} \Big)^{m + 1} \max \Big\{ \frac{1}{(1 + \log E(2 \Lambda))^m}, \frac{1}{E(2 \Lambda)} \Big\}.
\end{equation}
We now recall that
\begin{equation}
\label{def:BL}
B_\Lambda = \frac{1}{\pi} \int_0^\frac{\Lambda}{E(\Lambda)} \frac{z^2 - {z^4}/{3}}{1 - z^2} dz,
\end{equation}
so that, for $\Lambda \geq 1$,
$$B_\Lambda \leq \frac{2}{3 \pi} \int_0^\frac{\Lambda}{E(\Lambda)} \frac{dz}{1 - z} dz = \frac{2}{3 \pi} \log \left [ E(\Lambda) (\Lambda + E(\Lambda))\right ] \leq \frac{4}{3 \pi} \log E(2 \Lambda).$$
Combining with \eqref{valere}, we finally derive \eqref{estim:ULU}. 
We end the proof of Proposition \ref{Prop:ULU} by noting that \eqref{elise} follows directly from the definition of $U_\Lambda$ and \eqref{plaute}. \qed

\section{Proof of Proposition \ref{Prop:estimF}}\label{sec:proof_estim_F}
We may define $F_{n, \Lambda}^\epsilon(\mu)$ by duality as follows 
$$\int_{\R^3} \zeta F_{n, \Lambda}^\epsilon(\mu) = \tr(Q_{n, \Lambda}^\epsilon \zeta),$$
for any smooth function $\zeta$, and where
$$Q_{n, \Lambda}^\epsilon = \frac{1}{2 \pi} \int_{- \infty}^{\ii} \frac{1}{D^0 + i \eta} \prod_{j = 1}^n \Big( \Pi_\Lambda^{(\epsilon_j)} \mu_j \ast \frac{1}{|x|} \Pi_\Lambda^{(\epsilon_{j + 1})} \frac{1}{D^0 + i \eta} \Big) d\eta.$$
We will use, like in \cite[p. 547]{HaiLewSer-05a}, the inequality
\begin{equation}
\label{frosine}
|\tr(Q_{n, \Lambda}^\epsilon \zeta)| = \bigg| \int_{\R^3} \tr_{\C^4} \Big( \widehat{Q_{n, \Lambda}^\epsilon \zeta}(p, p) \Big) dp \bigg| \leq \int_{\R^3} \big| \widehat{Q_{n, \Lambda}^\epsilon \zeta}(p,p) \big| dp.
\end{equation}
The main idea is to derive a bound of the last integral in \eqref{frosine} in terms of the norms $\norm{(1 + \cU)^{- m} \zeta}_{L^2}$ and $\norm{(1 + \cU)^{- m} \zeta}_{L^2 + \cC'}$, which provides an estimate of the form \eqref{estim:F}, by duality. The proof will depend whether we estimate the integral in the right-hand side of \eqref{frosine} by the norm $\norm{(1 + \cU)^{- m} \zeta}_{\cC'}$ or by the norm $\norm{(1 + \cU)^{- m} \zeta}_{L^2}$. For this reason, we split it into three steps.

\setcounter{step}{0}
\begin{step}
\label{F1}
There exists a universal constant $C_1$ such that for all $n \geq 5$
\begin{equation}
\label{F:C5}
\norm{(1 + \cU)^m F_{n, \Lambda}^\epsilon(\mu)}_{\cC} \leq \frac{(C_1)^n  (K \log n)^m}{\Lambda^\frac{n(\epsilon)}{2}} \prod_{j = 1}^n \norm{(1 + \cU)^m \mu_j}_{\cC},
\end{equation}
for all $\mu = (\mu_1, \cdots, \mu_n) \in \cC^n$.
\end{step}

We estimate $\widehat{Q_{n, \Lambda}^\epsilon \zeta}(p, p)$ as follows:
\begin{multline}
\label{estim:rho_5_et_plus}
|\widehat{Q_{n, \Lambda}^\epsilon \zeta}(p, p)| \leq \frac{1}{(2 \pi)^{\frac{3 n + 5}{2}}} \int_{- \infty}^{\ii} d\eta \int_{\R^3} \cdots \int_{\R^3} \Big( f_\frac{1}{4}^{(\epsilon_1)}(p) |\widehat{\phi_1}(p - p_1)| f_\frac{1}{4}^{(\epsilon_2)}(p_1) \Big) \times\\
\times \prod_{j = 1}^{n - 1} \Big( f_\frac{1}{4}^{(\epsilon_{j + 1})}(p_j) |\widehat{\phi_{j + 1}}(p_j - p_{j + 1})| f_\frac{1}{4}^{(\epsilon_{j + 2})}(p_{j + 1}) \Big) \times\\
\times \Big( f_\frac{1}{4}^{(\epsilon_{n + 1})}(p_n) |\widehat{\zeta}(p_n - p)| f_\frac{1}{4}^{(\epsilon_1)}(p) \Big) dp_1 \cdots dp_{n},
\end{multline}
where $\phi_j = \mu_j \ast |\cdot|^{- 1}$, and for any $\beta > 0$, $f_\beta^{(\epsilon)} = \pi_\Lambda^{(\epsilon)}/(\eta^2 + E^2)^\beta$, with $\pi_\Lambda^{(\epsilon)} = 1$, if $\epsilon \neq - 1$, and $\pi_\Lambda^{(- 1)}(\cdot) = \1_{|\cdot| > \Lambda}$. Applying the following corollary of \eqref{U:prod}
$$(1 + U(p - p_n))^m \leq (K\log n)^m (1 + U(p - p_1))^m \underset{j = 1}{\overset{n - 1}{\prod}} (1 + U(p_j - p_{j + 1}))^m$$
to \eqref{estim:rho_5_et_plus}, we are led to estimating
\begin{multline*}
\int_{\R^3} |\widehat{Q_{n, \Lambda}^\epsilon \zeta}(p, p)| dp \leq \frac{1}{2 \pi}  (K\log n)^m \int_{- \infty}^{\ii} \tr \Big( \prod_{j = 1}^n \big( f_\frac{1}{4}^{(\epsilon_j)}(- i \nabla) \psi_j(x) \times\\
 \times f_\frac{1}{4}^{(\epsilon_{j + 1})}(- i \nabla) \big) \big( f_\frac{1}{4}^{(\epsilon_{n + 1})}(- i \nabla) \xi(x) f_\frac{1}{4}^{(\epsilon_1)}(- i \nabla) \big) \Big) d\eta,
\end{multline*}
where $\widehat{\psi_j} = (1 + U)^m |\widehat{\phi_j}|$ for $1 \leq j \leq n$, and $\widehat{\xi} = (1 + U)^{- m} |\widehat{\zeta}|$. Since $n + 1 \geq 6$, we deduce from H\"older's inequality in Schatten spaces \cite{Simon-79}, and the fact that $\| \cdot \|_{\gS_q} \leq \| \cdot \|_{\gS_r}$, as soon as $1 \leq r \leq q \leq \ii$, that
\begin{multline}
\label{estim:trace_Q_zeta}
\int_{\R^3} |\widehat{Q_{n, \Lambda}^\epsilon \zeta}(p, p)| dp \leq  \frac{1}{2 \pi} (K\log n)^m \int_{- \infty}^{\ii} \Big( \prod_{j = 1}^n \Big\| f_\frac{1}{4}^{(\epsilon_j)}(- i \nabla) \psi_j(x) \times\\
\times  f_\frac{1}{4}^{(\epsilon_{j + 1})}(- i \nabla) \Big\|_{\gS_6} \norm{f_\frac{1}{4}^{(\epsilon_{n + 1})}(- i \nabla) \xi(x) f_\frac{1}{4}^{(\epsilon_1)}(- i \nabla)}_{\gS_6} \Big) d\eta.
\end{multline}
We now use the Kato-Seiler-Simon inequality (see \cite{SeiSim-75} and \cite[Thm 4.1]{Simon-79}),
\begin{equation}
\label{KSS}
\forall p \geq 2,\qquad  \| f(- i \nabla) g(x) \|_{\gS_p} \leq \frac{1}{(2 \pi)^\frac{3}{p}}
\| g \|_{L^p(\R^3)} \| f \|_{L^p(\R^3)},
\end{equation}
to bound all the terms of the product in the right-hand side of \eqref{estim:trace_Q_zeta}. This provides
\begin{align*}
\norm{f_\frac{1}{4}^{(\epsilon)}(- i \nabla) |h(x)| f_\frac{1}{4}^{(\epsilon')}(- i \nabla)}_{\gS_6} \leq & \norm{f_\frac{1}{4}^{(\epsilon)}(- i \nabla) |h(x)|^\frac{1}{2}}_{\gS_{12}} \norm{|h(x)|^\frac{1}{2} f_\frac{1}{4}^{(\epsilon')}(- i \nabla)}_{\gS_{12}}\\
\leq & \frac{1}{(2\pi)^\frac{1}{2}} \norm{h}_{L^6} \norm{f_\frac{1}{4}^{(\epsilon)}}_{L^{12}} \norm{f_\frac{1}{4}^{(\epsilon')}}_{L^{12}},
\end{align*}
for any $\epsilon$ and $\epsilon'$ in $\{ - 1, 0, 1 \}$, and any $h \in L^6(\R^3)$. In particular, by the critical Sobolev inequality, we obtain for any function $h$ in $\dot{H}^1(\R^3)$,
\begin{equation}
\label{medecin}
\norm{f_\frac{1}{4}^{(\epsilon)}(- i \nabla) |h(x)| f_\frac{1}{4}^{(\epsilon')}(- i \nabla)}_{\gS_6} \leq A \norm{\nabla h}_{L^2} \norm{f_\frac{1}{4}^{(\epsilon)}}_{L^{12}} \norm{f_\frac{1}{4}^{(\epsilon')}}_{L^{12}},
\end{equation}
for some universal constant $A$. Given any $q \geq 2$ and $\beta > 6/q$, we then check that
\begin{align}
\label{argonte}
\norm{f_\beta^{(\epsilon)}}_{L^q} \leq E(\eta)^{\frac{3}{q} - 2 \beta} \bigg( \int_{\R^3} \frac{du}{E(u)^{\beta q}} \bigg)^\frac{1}{q} \leq & E(\eta)^{\frac{3}{q} - 2 \beta} \bigg( \int_{|u| \geq 1} \frac{du}{|u|^{2 \beta q}} \bigg)^\frac{1}{q}\\
= & \Big( \frac{4 \pi}{2 \beta q - 3} \Big)^\frac{1}{q} E(\eta)^{\frac{3}{q} - 2 \beta},\nonumber
\end{align}
for $\epsilon \neq - 1$, while similarly,
\begin{equation}
\label{celimene}
\norm{f_\beta^{(- 1)}}_{L^q} \leq \Big( \frac{4 \pi}{2 \beta q - 3} \Big)^\frac{1}{q} \min \Big\{ E(\eta)^{\frac{3}{q} - 2 \beta}, \Lambda^{\frac{3}{q} - 2 \beta} \Big\}.
\end{equation}
The definition of the functions $\psi_j$ gives
\begin{equation}
\label{jourdain}
\norm{(1 + \cU)^m \nabla \psi_j}_{L^2} = 4 \pi \norm{(1 + \cU)^m \mu_j}_{\cC},
\end{equation}
which, combined with \eqref{estim:trace_Q_zeta}, \eqref{medecin}, \eqref{argonte} and \eqref{celimene}, leads to
\begin{multline*}
\int_{\R^3} |\widehat{Q_{n, \Lambda}^\epsilon \zeta} (p, p)| dp \leq A^{n + 1} (K\log n)^m \norm{(1 + \cU)^{- m} \zeta}_{\cC'} \prod_{j = 1}^n \norm{(1 + \cU)^m \mu_j}_{\cC}\times\\
\times \int_0^{\ii} \min \Big\{ \frac{1}{E(\eta)^\frac{1}{2}}, \frac{1}{\Lambda^\frac{1}{2}} \Big\}^{n_\epsilon} \frac{d\eta}{E(\eta)^\frac{n + 1 - n_\epsilon}{2}},
\end{multline*}
for some universal constant $A$. When $n_\epsilon = 0$, we have
$\int_0^{\ii} {E(\eta)^{-(n + 1)/{2}}}\,{d\eta} \leq \int_0^{\ii} E(\eta)^{-3}\,{d\eta}$,
whereas, for $n_\epsilon = 1$,
\begin{equation*}
\int_0^{\ii} \min \Big\{ \frac{1}{E(\eta)^\frac{1}{2}}, \frac{1}{\Lambda^\frac{1}{2}} \Big\} \frac{d\eta}{E(\eta)^\frac{n}{2}} \leq  \frac{1}{\Lambda^\frac{1}{2}} \int_0^{\ii} \frac{d\eta}{E(\eta)^\frac{5}{2}} + \frac{1}{2 \Lambda^2}.
\end{equation*}
Inequality \eqref{F:C5} then follows with $C_1 =  {A}/{2} + A\int_0^{\ii} E(\eta)^{-5/2}d\eta$.

\begin{step}
\label{F2}
There exists a universal constant $C_2$ such that, for $n =3$ or $n \geq 5$,
\begin{equation}
\label{F:L2}
\norm{(1 + \cU)^m F_{n, \Lambda}^\epsilon(\mu)}_{L^2(\R^3)} \leq \frac{(C_2)^n (K\log n)^m}{\Lambda^\frac{n(\epsilon)}{7}} \prod_{j = 1}^n \norm{(1 + \cU)^m \mu_j}_{\cC},
\end{equation}
for all $\mu = (\mu_1, \cdots, \mu_n) \in \cC^n$.
\end{step}

The proof is similar to the proof of \eqref{F:C5}. Since $n/6 + 1/2 \geq 1$, we can now estimate $\widehat{Q_{n, \Lambda}^\epsilon \zeta}(p, p)$ by
\begin{multline}
\label{cleandre}
\int_{\R^3} |\widehat{Q_{n, \Lambda}^\epsilon \zeta}(p, p)| dp \leq \frac{(K\log n)^m}{2 \pi} \int_{- \infty}^{\ii} \Big( \norm{\psi_1(x) f_{\frac{3}{8}}^{(\epsilon_2)}(- i \nabla)}_{\gS_6} \times\\
\times \norm{f_{\frac{1}{8}}^{(\epsilon_2)}(- i \nabla) \psi_2(x) f_{\frac{1}{4}}^{(\epsilon_3)}(- i \nabla)}_{\gS_6} \prod_{j = 2}^{n - 2} \norm{f_{\frac{1}{4}}^{(\epsilon_{j + 1})}(- i \nabla) \psi_{j + 1}(x) f_{\frac{1}{4}}^{(\epsilon_{j + 2})}(- i \nabla)}_{\gS_6}\\
\times \norm{f_{\frac{1}{4}}^{(\epsilon_n)}(- i \nabla) \psi_n(x) f_{\frac{1}{8}}^{(\epsilon_{n + 1})}(- i \nabla)}_{\gS_6} \norm{f_{\frac{3}{8}}^{(\epsilon_{n + 1})}(- i \nabla) \xi(x) f_{\frac{1}{2}}^{(\epsilon_1)}(- i \nabla)}_{\gS_2} \Big) d\eta,
\end{multline}
where the functions $f_\beta^{(\epsilon)}$, $\psi_j$ and $\xi$ are defined as in Step \ref{F1}. Using H\"older's inequality and \eqref{KSS}, we can bound each norm in the right-hand side of \eqref{cleandre} similarly to \eqref{medecin}. This provides, for instance,
$$\norm{f_\frac{1}{8}^{(\epsilon_2)}(- i \nabla) |h(x)| f_\frac{1}{4}^{(\epsilon_3)}(- i \nabla)}_{\gS_6} \leq A \norm{\nabla h}_{L^2} \norm{f_\frac{1}{8}^{(\epsilon_2)}}_{L^{18}} \norm{f_\frac{1}{4}^{(\epsilon_3)}}_{L^{9}},$$
and
$$\norm{f_\frac{3}{8}^{(\epsilon_{n + 1})}(- i \nabla) |h(x)| f_\frac{1}{2}^{(\epsilon_1)}(- i \nabla)}_{\gS_2} \leq \frac{1}{(2 \pi)^\frac{3}{2}} \norm{h}_{L^2} \norm{f_\frac{3}{8}^{(\epsilon_{n + 1})}}_{L^\frac{14}{3}} \norm{f_\frac{1}{2}^{(\epsilon_1)}}_{L^\frac{7}{2}}.$$
Combining with \eqref{argonte}, \eqref{celimene} and \eqref{jourdain}, we obtain
\begin{multline*}
\int_{\R^3} |\widehat{Q_{n, \Lambda}^\epsilon \zeta}(p, p)| dp  \leq A^{n + 1} (K\log n)^m \norm{(1 + \cU)^{- m} \zeta}_{L^2} \prod_{j = 1}^n \norm{(1 + \cU)^m \mu_j}_{\cC} \times\\
 \times \int_0^{\ii} \min \Big\{ \frac{1}{E(\eta)}, \frac{1}{\Lambda} \Big\}^\frac{n_\epsilon}{7} \frac{d\eta}{E(\eta)^{\frac{3 n - 2}{6} + \frac{1 - n_\epsilon}{7}}},
\end{multline*}
for a universal constant $A$. Since
$$\int_0^{\ii} \frac{d\eta}{E(\eta)^{\frac{3 n - 2}{6} + \frac{1}{7}}} \leq \int_0^{\ii} \frac{d\eta}{E(\eta)^\frac{7}{6}},$$
for $n_\epsilon = 0$, whereas, for $n_\epsilon = 1$,
\begin{align*}
\int_0^{\ii} \min \Big\{ \frac{1}{E(\eta)^\frac{1}{7}}, \frac{1}{\Lambda^\frac{1}{7}} \Big\} \frac{d\eta}{E(\eta)^\frac{3 n - 2}{6}} \leq & \frac{1}{\Lambda^\frac{1}{7}} \int_0^{\ii} \frac{d\eta}{E(\eta)^\frac{7}{6}} + \frac{6}{\Lambda^\frac{1}{6}},
\end{align*}
we obtain \eqref{F:L2} with $C_2 = 6A + A\int_0^{\ii} E(\eta)^{-7/6}d\eta$.

\begin{step}
\label{F3}
Let $n =3$. There exists a universal constant $C_3$ such that
\begin{equation}
\label{F:C3}
\norm{(1 + \cU)^m F_{3, \Lambda}^\epsilon(\mu)}_{\cC} \leq \frac{(C_3)^3 (K \log 3)^m}{\Lambda^\frac{n(\epsilon)}{24}} \prod_{j = 1}^3 \norm{(1 + \cU)^m \mu_j}_{\cC},
\end{equation}
for all $\mu = (\mu_1, \mu_2, \mu_3) \in \cC^3$.
\end{step}

The proof of \eqref{F:C3} follows ideas from \cite[Section 4.3.4]{HaiLewSer-05a}. Contrarily to Steps \ref{F1} and \ref{F2}, it relies on an explicit computation of $\widehat{Q_{3, \Lambda}^\epsilon \zeta}(p, p)$ by means of the residuum formula for the integral with respect to the variable $\eta$. Indeed, it holds
\begin{equation}
\label{fourberies}
\widehat{Q_{3, \Lambda}^\epsilon \zeta}(p, p) = \sum_{\delta \in \{ - 1, 1 \}^4} \widehat{Q_{3, \Lambda}^{\epsilon, \delta} \zeta}(p, p).
\end{equation}
Here, the quantity $\widehat{Q_{3, \Lambda}^{\epsilon, \delta} \zeta}(p, p)$ vanishes if $\delta = \pm(1, 1, 1, 1)$, whereas, when $\delta = (1, - 1, - 1, - 1)$, it refers to the expression
\begin{multline*}
\widehat{Q_{3, \Lambda}^{\epsilon, \delta} \zeta}(p, p) = \frac{1}{(2 \pi)^6} \int_{\R^3} \int_{\R^3} \int_{\R^3} \pi_\Lambda^{(\epsilon_1)}(p) P ^0_+(p) \frac{\widehat{\phi_1}(p - p_1)}{E(p) + E(p_1)} \pi_\Lambda^{(\epsilon_2)}(p_1) P ^0_-(p_1) \times\\
\times\frac{\widehat{\phi_2}(p_1 - p_2)}{E(p) + E(p_2)} \pi_\Lambda^{(\epsilon_3)}(p_2) P ^0_-(p_2) \frac{\widehat{\phi_3}(p_2 - p_3)}{E(p) + E(p_3)} \pi_\Lambda^{(\epsilon_4)}(p_3) P ^0_-(p_3) \widehat{\zeta}(p_3 - p) dp_1\, dp_2\, dp_3,
\end{multline*}
where $P ^0_\pm(p) =(E(p) \pm (\alp.p + \beta))/2 E(p)$.
The expression of $\widehat{Q_{3, \Lambda}^{\epsilon, \delta} \zeta}(p, p)$ is similar when $\delta$ contains exactly one $\delta_i=1$, respectively exactly one $\delta_i=-1$. On the other hand, for $\delta = (1, 1, - 1, - 1)$, the function $\widehat{Q_{3, \Lambda}^{\epsilon, \delta} \zeta}(p, p)$ is given by
\begin{multline*}
\widehat{Q_{3, \Lambda}^{\epsilon, \delta} \zeta}(p, p)  = \frac{1}{(2 \pi)^6} \int_{\R^3} \int_{\R^3} \int_{\R^3} \pi_\Lambda^{(\epsilon_1)}(p) P ^0_+(p) |\widehat{\phi_1}(p - p_1)| \pi_\Lambda^{(\epsilon_2)}(p_1) P ^0_-(p_1) \times\\
\times  |\widehat{\phi_2}(p_1 - p_2)| \pi_\Lambda^{(\epsilon_3)}(p_2) P ^0_-(p_2) |\widehat{\phi_3}(p_2 - p_3)| \pi_\Lambda^{(\epsilon_4)}(p_3) P ^0_-(p_3) \widehat{\zeta}(p_3 - p) \times\\
 \times\Bigg( \frac{1}{(E(p) + E(p_2))(E(p_1) + E(p_2))(E(p_1) + E(p_3))}\\
 + \frac{1}{(E(p) + E(p_2))(E(p) + E(p_3))(E(p_1) + E(p_3))} \Bigg)\, dp_1\, dp_2\, dp_3.
\end{multline*}

We next estimate  $\widehat{Q_{3, \Lambda}^{\epsilon, \delta} \zeta}(p, p)$ as above. For instance, when $\delta = (1, - 1, - 1 , - 1)$, since $E(p + q) \leq E(p) + E(q)$ for any $(p, q) \in\R^3$, we can compute
\begin{align*}
\big| & \widehat{Q_{3, \Lambda}^{\epsilon, \delta} \zeta}(p, p) \big| \leq \frac{1}{(2 \pi)^6} \int_{\R^3} \int_{\R^3} \int_{\R^3} dp_1 dp_2 dp_3 \frac{|P ^0_+(p) \widehat{\phi_1}(p - p_1) P ^0_-(p_1)|}{E(p + p_1)^\frac{2}{3}} \times\\
& \times \pi_\Lambda^{(\epsilon_2)}(p_1) \frac{|\widehat{\phi_2}(p_1 - p_2)|}{E(p_1)^\frac{1}{6} E(p_2)^\frac{1}{2}} \pi_\Lambda^{(\epsilon_3)}(p_2) \frac{|\widehat{\phi_3}(p_2 - p_3)|}{E(p_2)^\frac{1}{2} E(p_3)^\frac{1}{2}} \pi_\Lambda^{(\epsilon_4)}(p_3) \frac{|\widehat{\zeta}(p_3 - p)|}{E(p_3)^\frac{1}{2} E(p)^\frac{1}{6}} \pi_\Lambda^{(\epsilon_1)}(p),
\end{align*}
so that, using \eqref{U:prod} as above,
\begin{equation}
\label{appri}
\begin{split}
\int_{\R^3} & \big| \widehat{Q_{3, \Lambda}^{\epsilon, \delta} \zeta}(p, p) \big| dp \leq (K \log 3 )^m \norm{M_{\phi_1}}_{\gS_2} \norm{\frac{\pi_\Lambda^{(\epsilon_2)}}{E^\frac{1}{6}} \big( - i \nabla \big) \psi_2(x) \frac{\pi_\Lambda^{(\epsilon_3)}}{E^\frac{1}{2}} \big( - i \nabla \big)}_{\gS_6}\\
& \times \norm{\frac{\pi_\Lambda^{(\epsilon_3)}}{E^\frac{1}{2}} \big( - i \nabla \big) \psi_3(x) \frac{\pi_\Lambda^{(\epsilon_4)}}{E^\frac{1}{2}} \big( - i \nabla \big)}_{\gS_6} \norm{\frac{\pi_\Lambda^{(\epsilon_4)}}{E^\frac{1}{2}} \big( - i \nabla \big) \zeta(x) \frac{\pi_\Lambda^{(\epsilon_1)}}{E^\frac{1}{6}} \big( - i \nabla \big)}_{\gS_6},
\end{split}
\end{equation}
where
$$\widehat{M_{\phi_1}}(p, q) = \frac{|\widehat{\phi_1}(p - q)| }{E(p + q)^\frac{2}{3}} |P ^0_+(p) P ^0_-(q)|.$$
The operator $M_{\phi_1}$ was estimated in Lemma 14 of \cite{HaiLewSer-05a} by
$\norm{M_{\phi_1}}_{\gS_2} \leq A \| \nabla \phi_1 \|_{L^2}$,
where $A$ is some universal constant. By the definition of $\phi_1$, we obtain
\begin{equation}
\label{precieuses}
\norm{M_{\phi_1}}_{\gS_2} \leq 4 \pi A \| (1 + \cU)^m \mu_1 \|_{\cC}.
\end{equation}
As for the other terms in the right-hand side of \eqref{appri}, we argue as before and, by \eqref{appri} and \eqref{precieuses}, we obtain
\begin{equation}
\label{voisee}
\begin{split}
\int_{\R^3} & \big| \widehat{Q_{3, \Lambda}^{\epsilon, \delta} \zeta}(p, p) \big| dp \leq \frac{A (K\log 3)^m}{\Lambda^\frac{n(\epsilon)}{24}} \| (1 + \cU)^{- m} \zeta \|_{\cC'} \prod_{j = 1}^3 \| (1 + \cU)^m \mu_j \|_{\cC},
\end{split}
\end{equation}
where $A$ denotes some universal constant. All the terms in the right-hand side of \eqref{fourberies} are similar to the one corresponding to $\delta = (1, - 1, - 1, - 1)$. In particular, \eqref{voisee} holds when $Q_{3, \Lambda}^{\epsilon, \delta} \zeta$ is replaced by $Q_{3, \Lambda}^\epsilon \zeta$. By duality, this completes the proofs of \eqref{F:C3} and of Step \ref{F3}.
This ends the proof of Proposition \ref{Prop:estimF}.\qed


\end{document}